\documentclass[10pt]{article}
\usepackage{microtype}
\usepackage{authblk}
\usepackage[letterpaper,margin=1in]{geometry}
\usepackage{amsmath,amsthm,amssymb,thmtools,thm-restate}
\usepackage{mathtools}
\usepackage{bm,xspace}
\usepackage[ruled,vlined,linesnumbered,commentsnumbered]{algorithm2e}
\usepackage{multirow}
\usepackage{booktabs}
\usepackage{enumitem}
\usepackage{tikz}
\usepackage[font=small,skip=0pt]{caption}

\newcommand{\posints}{\ensuremath{\mathbb{Z}_+}}
\newcommand{\reals}{\ensuremath{\mathbb{R}}}

\newcommand{\posreals}{\ensuremath{\mathbb{R}_+}}

\newtheorem*{lemma*}{Lemma}
\newtheorem*{theorem*}{Theorem}

\newcommand{\nn}{\mathbf{n}}
\newcommand{\na}{\mathbf{a}}
\newcommand{\nb}{\mathbf{b}}
\newcommand{\nc}{\mathbf{c}}
\newcommand{\nf}{\mathbf{f}}

\newtheorem{theorem}{Theorem}[section]
\newtheorem{lemma}[theorem]{Lemma}

\theoremstyle{definition}
\newtheorem{definition}[theorem]{Definition}
\newtheorem{proposition}[theorem]{Proposition}

\newcommand{\opt}{\mathit{OPT}}

\newcommand{\I}{\mathcal{I}}

\DeclareMathOperator{\rank}{rank}
\DeclareMathOperator{\spa}{span}
\DeclareMathOperator{\poly}{poly}
\newcommand{\outp}{R}

\newcommand{\lb}{\!\left(}
\newcommand{\rb}{\right)}

\DeclareMathOperator*{\argmax}{arg\,max}

\newcommand{\cI}{\mathcal{I}}
\newcommand{\cB}{B}
\newcommand{\cM}{\mathcal{M}}
\newcommand{\cH}{\mathcal{H}}
\newcommand{\cT}{\mathcal{T}}
\newcommand{\cO}{\mathcal{O}}

\DeclareMathAlphabet{\mathsfit}{T1}{\sfdefault}{\mddefault}{\sldefault}
\SetMathAlphabet{\mathsfit}{bold}{T1}{\sfdefault}{\bfdefault}{\sldefault}
\newcommand{\pU}{\mathsfit{U}}
\newcommand{\pT}{\mathsfit{T}}
\newcommand{\pP}{\mathsfit{P}}
\newcommand{\pQ}{\mathsfit{Q}}
\newcommand{\pA}{\mathsfit{A}}
\newcommand{\pp}{\mathsfit{p}}
\newcommand{\points}{\pP}
\newcommand{\pw}{\mathsfit{w}}
\newcommand{\pZ}{\mathsfit{Z}}
\newcommand{\pa}{\mathsfit{a}}
\newcommand{\pb}{\mathsfit{b}}
\newcommand{\pc}{\mathsfit{c}}
\newcommand{\pd}{\mathsfit{d}}
\newcommand{\pf}{\mathsfit{f}}
\newcommand{\pg}{\mathsfit{g}}

\newcommand{\Guess}{\textsc{Guess}}
\newcommand{\RepSet}{\textsc{RepSet}}
\newcommand{\SRepSet}{\textsc{StreamingRepSet}}

\newcommand{\StreamingCoverage}{\textsc{StreamingCoverage}}
\newcommand{\ProcessElem}{\textsc{ProcessElem}}
\newcommand{\FindNode}{\textsc{FindNode}}
\newcommand{\all}{\textsc{AllReps}}
\newcommand{\parent}{\textsc{ParentElem}}

\newcommand{\treeNodeThree}[9]{%
  \fill[#7] (#1,#2 - 0.66) circle[radius=2pt] node [above,#7] {\scriptsize  \ensuremath{#3}};
  \fill[#8] (#1+0.5,#2-0.66) circle[radius=2pt] node [above,#8] {\scriptsize \ensuremath{#4}};
  \fill[#9] (#1+1,#2-0.66) circle[radius=2pt] node [above,#9] {\scriptsize \ensuremath{#5}};
  \draw (#1,#2) arc[radius=0.5, start angle=90,end angle=270];
  \draw (#1,#2-1) -- (#1+1,#2-1);
  \draw (#1,#2) -- (#1+1,#2);
  \draw (#1+1,#2) arc[radius=0.5, start angle=90,end angle=-90];
  \draw (#1 -0.4 ,#2 - 0.5
) node [left] {\scriptsize \ensuremath{#6}};
}

\SetArgSty{textnormal}
\SetKwProg{myproc}{procedure}{}{}
\SetKwFor{ForEach}{for each}{do}{endfch}
\SetKwComment{myc}{$\triangleright$\ }{}
\SetCommentSty{textit}
\SetInd{0.5em}{0.5em}

\title{FPT-Algorithms for the $\ell$-Matchoid Problem with a Coverage Objective}

\author{}
\author[1]{Chien-Chung Huang}
\author[2]{Justin Ward}
\affil[1]{CNRS, DI ENS, PSL, France. \texttt{villars@gmail.com}}
\affil[2]{School of Mathematical Sciences, Queen Mary University of London, United Kingdom. \texttt{justin.ward@qmul.ac.uk}}

\begin{document}
\date{}
\maketitle
\begin{abstract}

We consider the problem of optimizing a coverage function under a $\ell$-matchoid of rank $k$. 
We design fixed-parameter algorithms as well as streaming algorithms to compute an exact solution. 
Unlike previous work that presumes linear representativity of matroids, we consider the general 
oracle model. 

For the special case where the coverage function is linear, we give a deterministic fixed-parameter 
algorithm parameterized by $\ell$ and $k$. This result, combined with the lower bounds of 
Lov\'asz~\cite{Lovasz1981}, and Jensen and Korte~\cite{Jensen82}, 
demonstrates a separation between the $\ell$-matchoid and the matroid $\ell$-parity problems 
in the setting of fixed-parameter tractability. 

For a general coverage function, we give both deterministic and randomized fixed-parameter algorithms, 
parameterized by $\ell$ and $z$, where $z$ is the number of points covered in an optimal solution. The resulting algorithms can be directly translated into streaming algorithms. For unweighted coverage functions, we show that we can find an exact solution even when the function is given in the form of a value oracle (and so we do not have access to an explicit representation of the set system). Our result can be implemented in the streaming setting and stores a number of elements depending only on $\ell$ and $z$, but completely indpendent of the total size $n$ of the ground set. This shows that it is possible to circumvent the recent space lower bound of Feldman et al.~\cite{feldman2020}, by 
parameterizing the solution value. This result, combined with existing lower bounds, also provides a new separation 
between the space and time complexity of maximizing an arbitrary submodular function and a coverage function 
in the value oracle model.
\end{abstract}

\newpage

\section{Introduction}
\label{sec:introduction}

A (weighted) coverage function $f : 2^X \to \posreals$ is defined by a collection $X$ of subsets of points\footnote{Here and throughout, we use the term ``points'' when discussing elements of the underlying universe of a coverage function to avoid confusion with the elements of $X$.} from some underlying universe, each with a weight. Given some $A \subseteq X$, $f(A)$ is simply the total weight of all points that appear in at least one set in $A$.
Here we consider the problem of maximizing a coverage function subject to one or more matroid constraints, which are captured by the notion of an $\ell$-matchoid. Formally, suppose we are given a ground set $X$ and a coverage function $f: 2^X \rightarrow \posreals$. 
The goal is to compute a \emph{feasible} set $S \subseteq X$ with $f(S)$ being maximized. 
The feasible sets of $X$ are defined by an $\ell$-matchoid $\cM$ over $X$. 
The $\ell$-matchoid $\cM$ is a collection $\{M_i=(X_i,\I_i)\}_{i=1}^s$ of matroids, each defined on some (possibly distinct) subset $X_i \subseteq X$, in which each element $e \in X$ appears in at most $\ell$ of the sets $X_i$. 
We then say that a set $S \subseteq X$ is feasible if and only if $S \cap X_i \in \I_i$ for each $1 \leq i \leq s$. 
Intuitively, an $\ell$-matchoid can be regarded as the intersection of several matroid constraints, in which any element ``participates'' in at most $\ell$ of the constraints. The \emph{rank} of an $\ell$-matchoid $\cM$, is defined as the maximum size of any feasible set. 
When the coverage function $f$ is an unweighted linear function, our problem is usually called the \textsc{$\ell$-Matchoid} problem in the literature~\cite{JenkynsMatchoid}. 

The family of $\ell$-matchoid constraints includes several other commonly studied matroid constraints. A $1$-matchoid is simply a matroid, and for $\ell > 1$, letting $s = \ell$ and $X_i = X$ for all $i$ gives an intersection of $\ell$ matroid constraints. Additionally, the $\ell$-set packing or $\ell$-uniform hypergraph matching problems can be captured by letting the elements of $X$ correspond to the given sets or hyperedges (each of which contains at most $\ell$ vertices) and defining one uniform matroid of rank 1 for each vertex, allowing at most $1$ hyperedge containing that vertex to be selected. It is NP-hard to maximize a coverage function even under a single, uniform matroid constraint, and the $\ell$-\textsc{Matchoid} problem (in our setting, this is the special case in which $f$ is an unweighted, linear function) is NP-hard when $\ell \geq 3$.
Thus, various approximation algorithms have been introduced for both coverage functions and, more generally, submodular objective functions in a variety of special cases e.g.,~\cite{Calinescu2011,Filmus2014,Fisher1978,Lee2013,Lee2010a,Nemhauser1978,Wolsey1982}. 

In this work, we study the problem from the point of view of 
\emph{fixed-parameter tractability}, in which some underlying parameter of problem instances is assumed to be a fixed-constant. A variety of fixed-parameter algorithms have been obtained for matroid constrained optimization problems, under the assumption that 
the matroids have a linear representation. Another approach is to require only that we are able to test whether or not any given set is feasible for each matroid (i.e.\ \emph{independence oracle}). This is typically the case in the \emph{streaming} setting, in which the  ground set $X$ is not known in advance but instead arrives one element at a time.  In this setting the algorithm may only store a small number of elements throughout its execution but must produce a solution for the entire instance at the end of the stream. Motivated by such settings, we consider what can be accomplished for such problems in the oracle model without access to a linear representation of the entire matroid.



\subsection{Our Contributions}
\label{sec:our-contributions}

We give FPT-algorithms for maximizing a coverage function $f$ under an $\ell$-matchoid $\cM$ of rank $k$, 
given only independence oracles for the matroids in $\cM$. Here the coverage function $f$ can be either given in the form of a value oracle, or explicitly as a family of sets over points. 
We accomplish our goal by constructing a \emph{joint $k$-representative set} for $\cM=\{M_i=(X_i,\I_i)\}_{i=1}^s$ with respect to a subset $T \subseteq X$. This is a set $R \subseteq T$ with the property that given a feasible set $B$ of $\cM$ with $|B| \leq k$ and any $e \in T \cap B$, there is some representative $e' \in R$ for $e$ so that $B-e+e'$ remains feasible in $\cM$. 
Our construction also works in the case in which each element $e \in X$ has some weight $w(e) \in \reals$, in which case we guarantee that an element $e$'s representative $e'$ has $w(e') \geq w(e)$.  Note that the set $R$ contains representatives only for those elements in $T$, but we ensure that these representatives provide valid exchanges with respect to \emph{any} sets $B_i \in \I_i$, which may include elements not in $T$. This allows us to easily employ our construction in the streaming setting, in which we can treat $T \subseteq X$ as the set of elements that is currently available to the algorithm at some time. 
Table~\ref{tab:1} gives a summary of our results. We emphasize 
that if we use any \emph{strict} subset of the parameters proposed, the problems become at least $W[1]$-hard---see Table~\ref{tab:2} for a summary and Appendix~\ref{sec:hardn-results-altern} for further details.

\begin{table}
\centering
\setlength{\tabcolsep}{0.25em}
\begin{tabular}{lcl@{\hskip 1em}lcc}
\toprule
Objective & Params & Kernel Size $(\ell = 1)$ & Kernel Size $(\ell > 1)$ & Type & Theorem \\
\midrule
Linear & $\ell,k$ & $k$ & $\cO(\ell^{(k-1)\ell})$ & D & Thm~\ref{thm:main-streaming}\\[0.5ex]
Unweighted Coverage (Oracle) & $\ell,z$ & 
$\cO(2^{(z-1)^2}z^{2z+1})$ &
$\cO(2^{(z-1)^2}\ell^{z(z-1)\ell}z^{z+1})$ & D & Thm~\ref{thm:main-oracle} \\
Weighted Coverage (Explicit) & $\ell,z$ & 
$(4e)^z\ln(\epsilon^{-1})$ &
$\cO((4e)^z\ell^{(z-1)\ell}\ln(\epsilon^{-1}))$ & R & Thm~\ref{thm:streaming-coverage} \\
Weighted Coverage (Explicit) & $\ell,z$ & 
$2^{\cO(z)}z\log^2(m)$ &
$2^{\cO(z)}\ell^{(z-1)\ell}\log^2(m)$ & D & Thm~\ref{thm:streaming-coverage} \\


 \bottomrule\\
\end{tabular}
\caption{A summary of our results.  All problems are constrained by an $\ell$-matchoid $\cM$ of rank $k$, and for coverage problems $z$ denotes the number of points covered in some optimal solution. The kernel size is stated as number of elements. In the last row, $m$ refers to the size of the underlying universe of the coverage function (i.e., the number of points). 
In the second last column, D indicates a deterministic algorithm and R a randomized algorithm with success probability $(1-\epsilon)$. In the offline setting, our algorithms require a number of independence oracle queries at most $n$ times the stated bounds on the kernel size.
}
\label{tab:1}
\end{table}

\begin{table}
\setlength{\tabcolsep}{0.5em}
\centering
\begin{tabular}{lccc}
 \toprule
 Objective & Params & Hardness  & Source \\
 \midrule
 Linear & $\ell$ &Para-NP-hard & \cite{DBLP:conf/coco/Karp72} \\
  Linear & $k$ &W[1]-hard & \cite{10.5555/2568438} \\
  Unweighted Coverage (Explicit) & $\ell, k$ & W[2]-hard & \cite{DBLP:journals/ita/BonnetPS16} \\
Unweighted Coverage (Explicit) & $z$ & W[1]-hard & \cite{10.5555/2568438} \\
  General Submodular (Oracle) & $\ell,k,f(\opt)$ & $\max\left(\Omega(n^k), \Omega(n^{f(\opt)/2})\right)$ queries &
 \cite{Jensen82,Lovasz1981} \\
\bottomrule\\
\end{tabular}
\caption{Hardness results for subsets of the parameters we consider. Here, $\ell$ is the number of matroids defining our $\ell$-matchoid and $k$ is the size of the solution. For coverage functions, $z$ is the number of points to cover. For the first four hardness results, see Appendix~\ref{sec:hardn-results-altern} for the reductions. For the last result, see discussion below.}
\label{tab:2}
\end{table}


As a warm-up, we show that a simple, combinatorial branching procedure is sufficient to produce a kernel  
of size $\Gamma_{\ell,k} \triangleq \sum_{q = 0}^{(k-1)\ell}\ell^q$ for the general, weighted \textsc{$\ell$-Matchoid} problem, parameterized by $\ell$ and the rank $k$ of the $\ell$-matchoid. This set can be computed \emph{deterministically} using $\Gamma_{\ell,k}\cdot |X|$ independence oracle queries plus the time required to sort the elements of $X$ by weight. To see this result in a larger context, we point out that for fixed-parameterized tractability, the \textsc{$\ell$-Matchoid} problem
is in a sense the most general problem one can handle using only an independence oracle, as the ``equivalent'' 
\textsc{$\ell$-Matroid Parity} problem cannot be solved with such an oracle. 

More precisely, in the \textsc{$\ell$-Matroid Parity} 
problem, we are given disjoint blocks of $k$ elements whose union must be independent in a single matroid. Although both this problem and the \textsc{$\ell$-Matchoid} are reducible to one another~\cite{Kaparis20,LovaszPlummer},
Lov\'asz~\cite{Lovasz1981} and Jensen and Korte~\cite{Jensen82} 
show that even when $\ell=2$, any algorithm finding $k$ blocks whose union is independent (implying a solution of value $2k$ in our setting) needs $\Omega(n^{k})$ independence queries. Therefore, our results give a new 
separation between the \textsc{$\ell$-Matchoid} and the \textsc{Matroid $\ell$-Parity} problems in the parameterized setting. To reconcile this apparent contradiction, we note that the classical reduction between these problems takes an instance of \textsc{Matroid $\ell$-Parity} with optimal solution $k$ to an \textsc{$\ell$-Matchoid} instance with optimal solution $n + k$. Thus, the given lower-bound for \textsc{Matroid $\ell$-Parity} indeed does not apply if we parameterize by $k$. It is also interesting to observe how critically the linear representability of matroids affects the tractability. Marx~\cite{DBLP:journals/tcs/Marx09} gives a randomized FPT-algorithm for \textsc{Matroid $\ell$-Parity}, parameterized by $\ell$ and $k$, thus showing that it is possible to circumvent the lower bound of~\cite{Jensen82,Lovasz1981} when the matroid is linear. 


Building on this, our main result considers the parameterized
\textsc{Maximum $(\cM,z)$-Coverage} problem, in which now $f$ is a general coverage function and we must select 
$S \subseteq X$ that is feasible for $\cM$ and covers either $z$ points or, in the weighted variant, $z$ points of maximum weight. We obtain FPT-algorithms for this problem, parameterized by $\ell$ and $z$ (Theorems~\ref{thm:main-oracle} and~\ref{thm:streaming-coverage}). Here, it is known that parameterizing by $\ell$ and $k$ causes the problem to be at least $W[2]$-hard. Coverage functions often serve as a motivating example for the study of submodular functions, and so it is tempting to ask whether one might obtain a similar FPT algorithm for an arbitrary submodular function by parameterizing by $z = f(\opt)$. However, this is impossible due to the aforementioned lower bound of~\cite{Jensen82,Lovasz1981}. To see this, observe that the objective for unweighted \textsc{Matroid 2-Parity} is a 2-polymatroid rank 
function, which is submodular. The lower bound construction of \cite{Jensen82,Lovasz1981} thus can be interpreted 
as follows: given a submodular function $f: 2^X \rightarrow \posints$, computing a set $S$ with $|S|\leq k$ and $f(S) \geq 2k$ requires $\Omega(n^{k})$ queries. In our context, by setting 
$z=2k$ implies that we need $\Omega(n^{z/2})$ queries even when $\cM$ is a single, uniform matroid.

In order to make this distinction rigorous, we again encounter questions of representation---for a general submodular function one must typically assume that the objective $f$ is available via a \emph{value oracle} that, for any set $S$, returns the value $f(S)$. In contrast, for coverage functions, $f$ can be given explicitly as a family of sets over the points in the universe, which may provide additional information not available in the value oracle model. In our main result, we show that it is possible to obtain a fixed-parameter algorithm for \textsc{Maximum $(\cM,z)$-Coverage} even when $f$ is given only as a value oracle, which reports only the number of points covered 
by a set $S \subseteq X$. This algorithm is technically the most demanding part of this paper. Here the lack of point representation requires 
a sophisticated data structure to store the elements properly. Moreover, we need to guarantee that the stored elements 
are not only compatible with the rest of the elements in the optimal solution under the $\ell$-matchoid, but 
also cover the points that are ``diffuse'' enough (so that at least one of them covers the points that are not 
already covered by the rest of the optimal solution). The latter goal is achieved by an extensive use of the 
joint $k$-representative sets. 

As a result, we demonstrate a new separation between what is possible for 
an arbitrary, integer-valued submodular function and an unweighted coverage function in the context of fixed-parameter tractability. A similar separation between coverage functions and arbitrary submodular functions was shown by Feige and Tennenholtz~\cite{Feige2017} 
and by Dughmi and Vondr\'ak~\cite{DBLP:journals/geb/DughmiV15} in different contexts. In the former, a separation of the 
approximability between a general and a submodular function is shown under a single uniform matroid, but applies only in a restricted setting in which the algorithm may only query the value of sets of size \emph{exactly} $k$. In the latter, 
a separation is established in the specific setting of truthful mechanism design. Our results imply a clean separation between what is possible for a coverage function and a general submodular function in the setting of fixed-parameter tractability, without 
any restriction on what types of sets the algorithm can query. Although our algorithm has rather high space and time complexity, 
here we emphasize that its main interest lies in the above theoretical implications. In Section~\ref{sec:impr-algor-expl}, we give more efficient randomized and deterministic algorithms via a color-coding technique when we have explicit access to the underlying representation of a coverage function. These algorithms work even in the general, weighted case.

All of our algorithms may be implemented in the streaming setting as well, which results in new consequences in the recently introduced setting of \emph{fixed-parameter streaming} 
algorithms~e.g., see \cite{DBLP:conf/iwpec/ChitnisC19,Chitnis2016,DBLP:conf/soda/ChitnisCHM15,DBLP:conf/mfcs/FafianieK14} and the references therein. Here, the idea is to allow the \emph{space} available for a streaming algorithm to scale as $g(p)\poly(\log n)$, where $p$ is a parameter. Just as the original motivation of fixed-parameterized complexity is
to identify the parameters that cause a problem to have large running time, here we want to identify the parameters that cause a problem to require large space. 

Recently, Feldman et al.~\cite{feldman2020} showed that given an unweighted coverage function $f$ and a uniform matroid of rank $k$ as contraint, a streaming algorithm attaining an approximation ratio of $1/2+\epsilon$ must use memory $\Omega(\epsilon n/k^3)$.
Our Theorem~\ref{thm:main-oracle} implies that one can circumvent this lower bound by parameterizing by the solution value $z$.  
As before, it is natural to ask when $f$ is an arbitrary submodular function in the value oracle model, can one solve the problem 
using the same amount of space? However, Huang et al.~\cite{DBLP:journals/corr/abs-2002-05477}, shows that to obtain the approximation ratio of $2 - \sqrt{2} + \epsilon$, one requires $\Omega(n/k^2)$ space, under 
a uniform matroid of rank $k$. This lower-bound construction uses a submodular function $f : 2^X \to \posints$ whose maximum value is $\cO(k^2)$. It implies that it is impossible to obtain an exact solution by storing only $g(z)$ elements, where $z=\cO(k^2)$, 
even for a single uniform matroid. Theorem~\ref{thm:main-oracle} thus 
again shows a separation in the space complexity between an arbitrary submodular function and a coverage function in the parameterized streaming setting. 

\subsection{Related Work}

Marx was the first to initiate the study of \textsc{Matroid $\ell$-Parity} from the perspective of fixed-parameterized 
tractability~\cite{DBLP:journals/tcs/Marx09}, using the idea of representative families. Fomin et al.~\cite{DBLP:journals/jacm/FominLPS16} gave an improved algorithm for constructing representative families that, when combined with the techniques from~\cite{DBLP:journals/tcs/Marx09}, leads to a randomized FPT algorithm for \emph{weighted} \textsc{Matroid $\ell$-Parity} in linear matroids. Their algorithm was subsequently derandomized by Lokshtanov et al.~\cite{DBLP:journals/talg/LokshtanovMPS18}, by showing that truncation can be performed on the representation of a linear matroid deterministically. 

The above results presume that the matroids in question are linearly representable. A related issue is how to compute the linear representation of such a matroid efficiently. Deterministic algorithms for finding linear representation of transversal matroids and gammoids are given by Misra et al.~\cite{MPRS2020} and Lokshtanov et al.~\cite{Lokshtanov2018}. The notion of 
\emph{union representation}, a generalization of linear representation, is also introduced in~\cite{Lokshtanov2018}.  

In another line of work, van Bevern et al.~\cite{DBLP:conf/ciac/BevernTZ19} considered a matroid constrained variant of facility location. Their approach can be shown to yield an FPT algorithm for weighted coverage functions subject to $\ell$ matroid constraints. 
Their algorithm requires linear representation for the underlying matroid when $\ell >1$. Moreover, as their approach is 
involved and uses an offline algorithm for 2-matroid intersection, it is unclear if it can be applied in the streaming setting without further insights.


The \textsc{Maximum $k$-Coverage} problem is an extensively studied special case of our problem, when $\cM$ is a single uniform matroid of rank $k$. Although this problem is known~\cite{Feige1998,Nemhauser1978} to be NP-hard to approximate beyond $(1-1/e)$, it is FPT when parameterized by the number $z$ of points to cover~\cite{DBLP:journals/ipl/Blaser03} or by the maximum of $k$ and the size of the largest set in $X$~\cite{DBLP:journals/ita/BonnetPS16}. 
However, it is $W[2]$-hard when parameterized by $k$ alone and $W[1]$-hard when parameterized by $k$ and the maximum number of sets any point appears in (sometimes called the \emph{frequency} of a point)~\cite{DBLP:journals/ita/BonnetPS16}, and FPT approximation schemes are known~\cite{DBLP:journals/iandc/Skowron17,DBLP:journals/jair/SkowronF17} when the maximum frequency is bounded. Recently, Manurangsi~\cite{DBLP:conf/soda/Manurangsi20} has shown that the problem cannot be approximated to 
better than $(1 - 1/e)$ in FPT time when parameterized by $k$, assuming the Gap-ETH. 

In the streaming setting, approximation algorithms for \textsc{Maximum $k$-Coverage} and more generally submodular optimization under special cases of the $\ell$-matchoid constraint were given in~\cite{Badanidiyuru2014a,DBLP:conf/spaa/BateniEM17,DBLP:journals/mp/ChakrabartiK15,DBLP:conf/icalp/ChekuriGQ15,DBLP:conf/nips/FeldmanK018,GJS2021,Huang20,DBLP:conf/icml/0001MZLK19,DBLP:journals/siamcomp/KorulaMZ18,LW2021,MTV2021,DBLP:journals/mst/McGregorV19,DBLP:conf/icml/Norouzi-FardTMZ18,DBLP:conf/sdm/SahaG09}. Recently, McGregor et al.~\cite{MTV2021} gave 
streaming exact and approximate algorithms for \textsc{Maximum $k$-Coverage}, as well as for the variant in which the goal is to maximize the number of points covered by \emph{exactly} one set (where, again, we may choose any collection of at most $k$ sets). Their algorithm stores $O(d^{d+1}k^d)$ elements, where $d$ is the maximum value of $f(e)$. For comparison, we parameterize by the total number of $z$ of points to be covered and use $2^{O(z)}$ space for a general matroid constraint (Thm~\ref{thm:streaming-coverage}) in the same explicit model.

\section{Preliminaries}
\label{sec:preliminaries}

Henceforth, we will use $A + e$ and $A - e$ to denote the sets $A \cup \{e\}$ and $A \setminus \{e\}$, respectively. For a set function $f : 2^X \to \mathbb{R}$, a set $A \subseteq X$ and an element $e \in X \setminus A$ we also use the shorthands $f(e)$ to denote $f(\{e\})$ and $f(e | A)$ to denote $f(A+e) - f(A)$.

A \emph{matroid} $M=(X,\I)$ over ground set $X$ is given by a family $\I \subseteq 2^X$ of \emph{independent sets} such that: (1) $\emptyset \in \I$, (2) $\I$ is downward closed: for all $A \subseteq B \subseteq X$, $B \in \I$ implies that $A \in \I$, and (3) $\I$ satisfies the augmentation property: if $A, B \in \I$ with $|A| < |B|$, there is some $e \in B \setminus A$ such that $A + e \in \I$. Here, we assume that matroids are given by an \emph{independence oracle}, which, when given a query set $A$, answers whether or not $A \in \I$.

For any set $A \subseteq X$, the \emph{rank} of $A$ in $M$ is given by $\rank_M(A) = \max\{|B|\,:\,B \subseteq A, B \in \I\}$. That is, $\rank_M(A)$ is the size of the largest independent set contained in $A$, and the rank of $M$ is simply $\rank_{M}(X)$, which is the common size of all maximal independent sets.

Here we will primarily work with the characterization of matroids in terms of \emph{spans}. Formally, the span of $A$ in $M$ is defined as $\spa_M(A) = \{e \in X\,:\,\rank_M(A +e) = \rank_M(A) \}$. Note that for any $T \subseteq X$, we have $T \subseteq \spa_M(T)$, and for independent $T \in \I$, $\rank_M(\spa_M(T)) = |T|$. Additionally, for $T \in \I$, $\spa_M(T) = T \cup \{e \in X \setminus T\,:\,T+e \not\in \mathcal{I}\}$. Thus, it is straightforward to compute the span of an independent set $T$ by using an independence oracle for $M$. The following additional facts will be useful in our analysis:
\begin{proposition} Let $M=(X,\I)$ be a matroid. Then,
\begin{enumerate}[topsep=0pt,itemsep=-1ex,partopsep=1ex,parsep=1ex]
\item For any sets $S,T \subseteq X$, if $S \subseteq \spa_M(T)$, then $\spa_M(S) \subseteq \spa_M(T)$. 
\item For any $S,T \in \I$ with $S \subseteq \spa_M(T)$ and $|S| = |T|$,  $\spa_M(S) = \spa_M(T)$.
\end{enumerate}
\label{prop:span}
\end{proposition}
\begin{proof}
The first claim is well-known (see e.g.~\cite[Theorem 39.9]{Schrijver2003}). For the second, note that since $S \subseteq \spa_M(T)$, we must have $\spa_M(S) \subseteq \spa_M(T)$ by the first claim. Suppose for the sake of contradiction that there is some element $e \in \spa_M(T) \setminus \spa_M(S)$. Then, $S+e \in \I$. Moreover, $S \subseteq \spa_M(S) \subseteq \spa_M(T)$, so $S+e \subseteq \spa_M(T)$. However, this means that $\spa_M(T)$ contains an independent set $S+e$
of size $|S|+1 = |T|+1$, and so $\rank_M(\spa_M(T)) \geq |T| + 1 > |T|$---a contradiction.
\end{proof}

Recall that we define an $\ell$-matchoid on $X$ as a collection $\cM=\{M_i=(X_i,\I_i)\}_{i = 1}^s$ of matroids, where each $X_i \subseteq X$ and any $e \in X$ appears in at most $\ell$ of the $X_i$. For every element $e \in X$, we let $X(e)$ denote the collection of the (at most $\ell$) ground sets $X_i$ with $e \in X_i$. We say that a set $S \subseteq X$ is \emph{feasible} for $\cM$ if $S \cap X_i \in \I_i$ for all $1 \leq i \leq s$. The \emph{rank} of an $\ell$-matchoid $\cM$ is the maximum size of a feasible set for $\cM$.
We suppose without loss of generality that for each element $e \in X$, $\{e\} \in \cI_i$ for all $X_i \in X(e)$, (i.e.\ none of the matroids in $\cM$ has loops), in other words, $\{e\}$ is feasible in $\cM$. Note that any element $e$ for which this is not the case cannot be part of any feasible solution and so can be discarded.

\section{Joint $k$-Representative Set}
\label{sec:our-main-constr}

Our main construction will involve the following notion of a representative set for a  collection $\cM = \{M_i=(X_i,\I_i)\}_{i = 1}^s$ of matroids.
\begin{definition}[Joint $k$-representative set for $(T,\cM,w)$]
\label{def:jointrepset}

Let $X$ be a set and suppose that each element $e \in X$ has some weight $w(e) \in \reals$. Let $\cM = \{M_i=(X_i,\I_i)\}_{i=1}^s$ be a matchoid with $X_i \subseteq X$ for all $1 \leq i \leq s$. Finally, let $T$ be a fixed subset of $X$. 

We say that some subset $R \subseteq T$ is a \emph{joint $k$-representative set for $(T,\cM,w)$} if for any feasible set $\cB$ of $\cM$, with $|\cB|\leq k$, and for any element $b \in T \cap \cB$, there exists some $e \in R$ with $w(e) \geq w(b)$ and 
$\cB-b+e$ feasible for $\cM$. 
\end{definition}

Note that such a joint $k$-representative set $R \subseteq T$ has the property that for any feasible solution $O$ of size at most $k$ in $\cM$, and each $b \in O \cap T$: either $R$ contains $b$ already (in which case, we let $e = b$), or $R$ contains some other element $e$ so that $O-b + e$ remains feasible and the new weight $w(O-b+e) \geq w(O)$. As we show in Theorem~\ref{thm:matchoid-offline}, if $R$ is a joint $k$-representative set for $X$, then it then follows that $R$ must contain a feasible solution of size at most $k$ with total weight at least as large as any feasible set $O \subseteq X$ with size at most $k$.


We now give an algorithm for computing a joint $k$-representative set for $(T,\cM,w)$. Our main procedure, $\RepSet$ is presented in Algorithm~\ref{alg:guess}. We suppose that we are given access to independence oracles for all matroids in $\cM$, as well as a weight function $w : X \to \reals$.
In order to compute a joint $k$-representative set for $(T,\cM,w)$,  the procedure $\RepSet(T)$ makes use of an auxiliary procedure $\Guess$ that takes a pair $(J,Y)$ as input. 
The first input $J$ to $\Guess$ is a multi-dimensional set $J=(J_1,\cdots, J_s)$, where each $J_i \subseteq X_i$ for $1 \leq i \leq s$, and the second input is a subset $Y \subseteq T$.
For a multi-dimensional set $J$, we define $\|J\| \triangleq \sum_{i=1}^{s}|J_i|$, and let $J +_i e$ denote the multi-dimensional set obtained from $J$ by adding $e$ to the set $J_i$. That is, $J +_{i} e = (J_1,\ldots, J_{i-1},J_i+e,J_{i+1},\ldots, J_s)$. Given a pair of inputs $(J, Y)$, the procedure $\Guess$ first selects a maximum weight element $e$ of $Y$ and adds $e$ to the output set $\outp$. If $\|J\| < (k-1)\ell$, it then considers each of the matroids $M_i = (X_i,\I_i)$ for which $X_i \in X(e)$. For each of these, it makes a recursive call in which $e$ has been added to the corresponding set $J_i$ of $J$ and all elements of $Y$ spanned by $J_i+e$ in $M_i$ have been removed from $Y$.

\begin{algorithm}[h]
\KwIn{parameters $k, \ell$, independence oracles for $\ell$-matchoid $\cM = \{M_i\}_{i=1}^s$ of rank $k$, weight function $w: X \to \reals$.}
\smallskip
\myproc{$\RepSet(T)$}{
  \Return the output of $\Guess((\emptyset,\ldots,\emptyset),T)$\;
}
\myproc{$\Guess(J=(J_1,\cdots, J_s),Y)$}{
\lIf{$Y = \emptyset$}{\Return $\emptyset$}
Let $e = \arg\max_{a \in Y} w(a)$\;
$\outp=\{e\}$\;
\If{$\|J\| < (k-1)\ell$}{
\ForEach{$X_i \in X(e)$}{
 Define $Y_i = Y \backslash \spa_{M_i}(J_i+e)$\label{li:guess-span}\;
 $\outp = \outp \cup \Guess(J +_{i} e , Y_i)$\label{li:guess-rec}\;
}

}
\Return $\outp$\;
}
\caption{FPT-algorithm}
\label{alg:guess}
\end{algorithm}

In our analysis, it will be helpful to consider the tree of recursive calls to $\Guess$ made during the execution of $\RepSet(T)$. Each node in this tree corresponds to some call $\Guess(J,Y)$, where $Y \subseteq T$ and $J$ is a multi-dimensional set. 

The following proposition is a rather straightforward consequence of our algorithm. 
\begin{proposition} For any call $\Guess(J,Y)$ in the tree of recursive calls made by $\RepSet(T)$,
\begin{enumerate}[topsep=0pt,itemsep=-1ex,partopsep=1ex,parsep=1ex]
\item $J_i \in \I_i$ for $1 \leq i \leq s$. 
\item $e \in Y$ if and only if $e \in T$ and for every $X_i \in X(e)$, $e \not \in \spa(J_i)$. 
\end{enumerate}
\label{pro:input}
\end{proposition}

\begin{proof} We prove the proposition to be true by induction on the depth of the tree node corresponding to a call $\Guess(J,Y)$.
In the root, the proposition holds trivially as $J= (\emptyset,\ldots,\emptyset)$. Consider now a non-root node corresponding 
to some call $\Guess(J,Y)$. Such a call is invoked by the parent node corresponding to some other call $\Guess(J',Y')$, 
where $J = J' +_{i} e$ for some $e \in Y'$ and $X_i \in X(e)$. 

For all $i' \neq i$, $J_{i'}=J'_{i'} \in \mathcal{I}_{i'}$, by induction hypothesis. For $J_i$, we note that by the induction hypothesis, $e \not \in \spa(J'_i)$ and $J'_i \in \cI_i$. Thus, $J_i = J'_i+e \in \mathcal{I}_i$ and part (1) of the proposition is proved. For part (2), by induction hypothesis, $e \in Y'$ if and only if $e \in T$ 
and for every $X_i \in X(e)$, $e \not \in \spa(J_i)$ and by the operation of the algorithm, $e \in Y' \backslash Y$ if and only if $e \in Y'$ and $e \in \spa_{M_i}(J'_i + e)$. The proof then follows.
\end{proof}
\noindent
Note that we can compute $Y \backslash \spa_M(J_i+e)$ in line~\ref{li:guess-span} of Algorithm~\ref{alg:guess} by using at most $|Y| \leq |T|$ independence oracle calls for $M_i$. Each call to $\RepSet$ will result in several recursive calls to $\Guess(J,Y)$. In our analysis, it will be useful to consider inputs $J,Y$ that satisfy the following property:

\begin{definition}
\label{def:legitimate}
Given a feasible set $\cB$ in the matchoid $\cM$ and $b \in \cB \cap T$, 
we call a pair of inputs $(J,Y)$ \emph{legitimate for $(\cB,b)$} 
if $b \in Y$ and $J_i \subseteq \spa_{M_i}(B_i - b)$, where $B_i = \cB \cap X_i$, for all $1 \leq i \leq s$. 
\end{definition}

\noindent Using this definition, we now formally analyze the behavior of our algorithm.

\begin{lemma} Suppose that the input $(J,Y)$ to the call $\Guess$ is legitimate for $(\cB, b)$. Consider the element $e = \arg\max_{a \in Y}w(a)$ selected in this call to \Guess. If $e \in \spa_{M_i}(B_i - b)$ for some $X_i \in X(e)$, then $|J_i| < |B_i - b|$.
\label{lem:eitherOr}
\end{lemma}

\begin{proof} Suppose that $e \in \spa_{M_i}(B_i - b)$ for some $X_i \in X(e)$ and assume for the sake of contradiction that $|J_i| \geq |B_i-b|$. By Proposition~\ref{pro:input}(1), $J_i \in \I_i$ and 
by definition $B_i - b \in \I_i$. Moreover, since $(J,Y)$ is legitimate for $(\cB,b)$, we have $J_i \subseteq \spa_{M_i}(B_i - b)$,
and, as $J_i$ and $B_i - b$ are independent, $|J_i| \leq |B_i - b|$. Thus, $|J_i| = |B_i-b|$.  Proposition~\ref{prop:span}(2) then implies that $\spa_{M_i}(B_i - b) = \spa_{M_i}(J_i)$. But then by Proposition~\ref{pro:input}(2), $e \not\in \spa_{M_i}(J_i)$ and so $e \not\in \spa_{M_i}(B_i - b)$---a contradiction. 
\end{proof}

\begin{lemma}\label{lem:claim-1}
Consider a legitimate input $(J,Y)$ for $(\cB,b)$, and let $e = \arg\max_{a \in Y}w(a)$. Then, $w(e) \geq w(b)$ and either $\cB-b+e$ remains feasible or there is some $X_i \in X(e)$ such that $J' = J +_i e$, $Y' = Y\setminus \spa_{M_i}(J_i+e)$ is a legitimate input for $(\cB,b)$, with $\|J'\| = \|J\| + 1$.
\end{lemma}
\begin{proof}
Since $(J,Y)$ is legitimate for $(\cB,b)$, we have $b \in Y$ and so $w(e) \geq w(b)$. If $\cB-b+e$ is not feasible, then 
$e \in \spa_{M_i}(B_i - b)$ for some $X_i \in X(e)$. 
By Lemma~\ref{lem:eitherOr}, $|J_i| < |B_i - b|$. 
Since $(J,Y)$ is legitimate for $(\cB,b)$, $J_i \subseteq \spa_{M_i}(B_i - b)$ and so in fact $J'_i=J_i + e \subseteq \spa_{M_i}(B_i - b)$, and for all $i' \neq i$, $J'_{i'} = J_{i'}  \subseteq \spa_{M_{i'}}(B_i-b)$. 

What remains to argue is that $b \in Y' = Y \setminus \spa_{M_i}(J_i+e)$. 
If $b \not \in B_i$, then $b \not \in X_i$, implying that $b$ cannot be part of  $\spa_{M_i}(J_i+e)$. So assume that 
$b \in B_i$. If $b \in \spa_{M_i}(J_i + e)$ then, since $J_i + e \subseteq \spa_{M_i}(B_i - b)$, Proposition~\ref{prop:span}(1) implies that $b \in \spa_{M_i}(B_i - b)$, contradicting that $B_i = \cB \cap X_i \in \mathcal{I}_i$. 
\end{proof}

We are now ready to prove our first main result: $\RepSet(T)$ constructs a joint $k$-representative set for $(T,\cM,w)$.

\begin{theorem} \label{thm:main-repset}
Consider an $\ell$-matchoid $\cM=\{M_i=(X_i, \I_i)\}_{i=1}^s$ and weight function $w : X \to \reals$. Then, for any subset $T \subseteq X$, $\RepSet(T)$ returns a joint $k$-representative set $\outp$ for $(T,\cM,w)$, with $|R| \leq \Gamma_{\ell,k} \triangleq \sum_{q = 0}^{(k-1)\ell}\ell^{q}$ using at most $\Gamma_{\ell,k}\cdot |T|$ independence oracle queries. For $\ell = 1$, $|R| \leq k$ and for $\ell > 1$, $|R| \leq \frac{\ell}{\ell-1}\ell^{(k-1)\ell} = \cO\lb\ell^{(k-1)\ell}\rb$.
\end{theorem}
\begin{proof}
We begin by showing that the set $\outp$ returned by $\RepSet(T)$ is a joint $k$-representative set for $(T,\cM,w)$. 
Let $\cB$ be a feasible set in $\cM$ and $|\cB|\leq k$. We need to show that for any $b \in T \cap \cB$, 
there must exist some $e \in \outp$ with $w(e) \geq w(b)$ and $\cB - b + e$ remains feasible in $\cM$. In the following we fix an arbitrary element $b \in T \cap \cB$.

First, we note that any $(J,Y)$ that is legitimate for $(\cB, b)$ must have $\|J\| \leq (k-1)\ell$. To see this, 
let $\cB^- = \cB - b$. Then $|\cB^-| \leq k - 1$. 
Since $\cM$ is an $\ell$-matchoid, each element in $\cB^-$ appears in at most $\ell$ of the sets $\cB^- \cap X_i$. 
Then, from Definition~\ref{def:legitimate} together with Proposition~\ref{pro:input}(1), a legitimate input $(J,Y)$ must have:
$\|J\| = \sum_{i = 1}^s|J_i| \leq \sum_{i = 1}^s|\cB^-\cap X_i| \leq \ell|\cB^-| \leq \ell(k - 1)$.

Now, we consider the set of all recursive calls made to $\Guess(J,Y)$ by $\RepSet(T)$. We first show that for any $0 \leq d \leq (k-1)\ell$, either an element $e$ with the desired properties is added to $\outp$ by call $\Guess(J,Y)$ with $\|J\| < d$ or there is some call to $\Guess(J,Y)$ with $(J,Y)$ legitimate for $b$ and $\|J\| = d$. We proceed by induction on $d$. Initially $\RepSet(T)$ makes a call to $((\emptyset,\ldots,\emptyset),Y=T)$, which is legitimate for $b$, since $b \in T$ by assumption. For the induction step, suppose that no call to $\Guess(J,Y)$ with $\|J\| < d < (k-1)\ell$ adds an element $e$ with the desired property to $\outp$. Then, by the induction hypothesis, there is some call to $\Guess(J,Y)$ with $(J,Y)$ legitimate for $b$ and $\|J\| = d$. Consider the element $e$ selected by this call. By Lemma~\ref{lem:claim-1}, either $e$ has the desired properties or there is some $X_i \in X(e)$ such that $J' = J +_i e$, $Y' = Y \setminus \spa_{M_i}(J_i + e)$ is legitimate for $(\cB,b)$ and $\|J'\| = \|J\| + 1 = d+1$. In the latter case, since $\|J\| < (k-1)\ell$, the procedure $\Guess(J,Y)$ will make a recursive call $\Guess(J',Y')$, for this legitimate input $(J',Y')$. This completes the proof of the induction step.

Suppose now, for the sake of contradiction, that no call to $\Guess(J,Y)$ made by $\RepSet(T)$ adds an element $e$ with the desired properties to $\outp$. Then by the claim above (with $d = (k-1)\ell$), there is some call to $\Guess(J,Y)$ with $(J,Y)$ legitimate for $(\cB,b)$ and $\|J\| = (k-1)\ell$. As this call must not have selected an element $e$ with the desired properties, Lemma~\ref{lem:claim-1} implies that there must be some $(J',Y')$ that is legitimate for $(\cB,b)$ with $\|J'\| = (k - 1)\ell + 1$, contradicting our bound on the size of any legitimate $\|J'\|$. Thus, for any arbitrary $b \in T \cap \cB$ there is indeed some $e \in \outp$ with $w(e) \geq w(b)$ and $\cB-b+e$ feasible for $\cM$ and so $\outp$ is then a joint $k$-representative set for $(T,\cM,w)$. 

Finally, we consider the complexity of the procedure $\RepSet(T)$. Consider the tree of recursive calls to $\Guess$ made by $\RepSet(T)$. Each call in this  tree contributes at most 1 additional element to the final output set. For all calls except the root, we also make at most $|Y| \leq |T|$ independence queries in Line~\ref{li:guess-span} of Algorithm~\ref{alg:guess} immediately before making this call. It follows that the total number of independence oracle queries is at most $|T|$ times the size of the recursion tree. Now, we note that each non-leaf call $\Guess(J,Y)$ of the tree has at most $\ell$ children $\Guess(J',Y')$ and for each child, $\|J'\| \geq \|J\| + 1$. Thus the depth of the recursion tree is at most $(k-1)\ell$ and so contains at most $\Gamma_{\ell,k} = \sum_{q = 0}^{(k-1)\ell} \ell^q$ calls. The stated bounds then follow.
\end{proof}
We now show that any joint $k$-representative set for $(X,\cM,w)$ can be used as a kernel for maximizing a linear function under an $\ell$-matchoid constraint $\cM$.

\begin{theorem}\label{thm:matchoid-offline}
Let $\cM=\{M_i\}_{i=1}^s$ be an $\ell$-matchoid of rank $k$. Then, the procedure $\RepSet(X)$ computes a kernel $R$ for finding a maximum weight feasible set for $\cM$ with $|R| \leq \Gamma_{\ell,k} \triangleq \sum_{q = 0}^{(k-1)\ell}\ell^q$. The procedure requires the time to make $\Gamma_{\ell,k}\cdot |X|$ independence oracle queries, plus the time required to sort the elements of $X$ by weight.
\end{theorem}
\begin{proof}
Suppose that we are given a feasible set $O = \{b_1,\ldots,b_{k'}\}$, where $k' \leq k$. We show by induction on $0 \leq r \leq k'$ that there is some set $S_r \subseteq \outp$ such that $O_r = O \setminus \{b_1,\ldots,b_r\} \cup S_r$ is feasible for $\cM$, $|O_r| = |O|$ and $w(O_r) \geq w(O)$. If $r = 0$, then the claim holds trivially with $S_0 = \emptyset$, and $O_0 = O$.

In the general case, suppose that $r > 0$. By the induction hypothesis, there is some set of elements $S_{r-1} \subseteq R$ such that $O_{r-1} = O \setminus \{b_1,\ldots,b_{r-1}\} \cup S_{r-1}$ is feasible for $\cM$, $|O_{r-1}| = |O|$, and $w(O_{r-1}) \geq w(O)$. 
By Theorem~\ref{thm:main-repset}, there is then some element $e_r$ in the output of $\RepSet(X)$ with 
$O_{r-1}-b_r + e_r$ feasible and $w(e_r) \geq w(b_r)$. Let $S_r = S_{r-1} + e_r$ so that $O_r = O \setminus \{b_1,\ldots,b_r\} \cup S_r = O_{r-1} - b_r + e_r$. Then, $|O_r| = |O_{r-1}| = |O|$ and $w(O_r) \geq w(O_{r-1}) \geq w(O)$.  


The bounds on the number of oracle queries follows directly from Theorem~\ref{thm:main-repset}. Additionally, we note that each call to $\Guess(J,Y)$ requires finding the maximum weight element $e \in Y$. This can be accomplished by sorting $X$ at the beginning of the algorithm, and then storing each $Y$ according to this sorted order.
\end{proof}



\subsection{Joint $k$-Representative Sets in the Streaming Setting}

\label{sec:streaming-algorithm}

We next show that joint $k$-representative sets in the preceding section can be implemented in the streaming setting. Here, we suppose that the elements of $X$ are initially unknown, and at each step a new element $e$ arrives in the stream, together with the indices of the ground sets $X_i \in X(e)$. Recall that we are parameterizing by $\ell$, so we can assume that $\ell \leq n$. Furthermore, since each element participates in at most $\ell$ sets $X_i$ of an $\ell$-matchoid $\cM = \{M_i =(X_i,\I_i)\}_{i=1}^s$, we can assume that $s \leq n\ell$.

Our algorithm, shown in Algorithm~\ref{alg:streaming-guess}, maintains a representative set for all the elements that have previously arrived. When a new element arrives, we show that a new representative set for the entire stream can be obtained by applying the procedure $\RepSet$ to the set $T$ containing the representative set for the elements that have previously arrived together with this new element.

\begin{theorem}
\label{thm:main-streaming}
Consider an $\ell$-matchoid $\cM=\{M_i=(X_i, \I_i)\}_{i=1}^s$ and weight function $w : X \to \reals$.  Then, the set $R$ produced Algorithm~\ref{alg:streaming-guess} is a joint $k$-representative set for $(T,\cM,w)$, where $T$ is the subset of $X$ arriving in the stream so far. $|R| \leq \Gamma_{\ell,k} \triangleq \sum_{q = 0}^{(k-1)\ell}\ell^q$ and at all times during its execution, processing the arrival of an additional element requires temporarily storing this element together with an additional $\cO(k\ell\log n)$ bits. For $\ell = 1$, $|R| \leq k$ and for $\ell > 1$, $|R| = \cO\lb \ell^{(k-1)\ell}\rb$.
\end{theorem}
\begin{proof}
We proceed by induction on the stream of elements, in order of arrival. Let $\cB$ be a feasible set in $\cM$. 
For each $0 \leq t \leq n$, let $A_t$ be the first $t$ elements that arrive in the stream and $\outp_{t-1}$ be the current set $\outp$ immediately before the $t$-th element arrives. We show by induction that for each $0 \leq t \leq |T|$, for any $b \in A_t \cap \cB$ there is some $e \in \outp_t$, such that $\cB - b +e$ is feasible and $w(e) \geq w(b)$.  For $t = 0$, we have $A_t = \emptyset$ and so the claim follows trivially.

Let $t > 0$ and consider the arrival of the $t$-th element $e_t$ in the stream. Then, $A_{t} = A_{t-1} + e_t$.  Fix any element $b \in A_t \cap \cB$. We consider first the case that $b \in A_{t-1} \cap \cB$. By the induction hypothesis there is some $e \in \outp_{t-1}$ with $\cB-b+e$ feasible and $w(e) \geq w(b)$. Let $\cB' = \cB -b + e$. By Theorem~\ref{thm:main-repset}, $\outp_{t}$ is a joint $k$-representative set for $(\outp_{t-1}+e_t,\cM,w)$. Then, since $e \in \outp_{t-1}$ there is some $e' \in \outp_{t} = \RepSet(\outp_{t-1} + e_t)$  such that $w(e') \geq w(e) \geq w(b)$ and $\cB'-e+e'=\cB-b+e'$ feasible. 

Next consider the case $b = e_{t}$. Again, since $\outp_{t}$ is a joint $k$-representative set for $(\outp_{t-1}+e_t)$ there must exist some $e' \in \outp_t$ with $w(e') \geq w(e_t) = w(b)$ and $\cB - e_t +e' = \cB - b + e'$ feasible. 
This completes the proof of the induction step. The first claim in the theorem then follows by letting $t = |T|$, and noting that $A_{|T|} = T$ and $\outp_{|T|}$ is the set $\outp$ at the moment all of $T$ have arrived.

We note that by Theorem~\ref{thm:main-repset}, the size of the set $\outp$ computed in any step of the algorithm is always at most $\Gamma_{\ell,k}$. In order to process the arrival of an element $e$, the algorithm computes $\RepSet(\outp + e)$. This makes a tree of recursive calls $\Guess(J,Y)$, where $J$ is a multidimensional set and $Y \subseteq \outp+e$. As shown in the proof of Theorem~\ref{thm:main-repset}, this tree has depth at most $(k-1)\ell$ and so at any time we must maintain at most $(k-1)\ell$ such inputs $(J,Y)$ appearing on the path from the current call to the root of the tree. To store each $J$, we note that each recursive call made by $\Guess(J,Y)$ adds some element $\overline{e} \in Y$ to a set $J_i \in J$. Thus, we can represent $J$ implicitly by storing $\overline{e}$, together with a currently selected index $i$ at each such call in the tree. Storing this index requires $\log(s) \leq \log(n\ell) = O(\log(n))$ bits. Moreover, given $J$, we can easily determine $Y$, since it is precisely the set of elements $e' \in \outp+e$ such that $e' \not\in \spa_{M_i}(J_i)$ for all $X_i \in X(e')$.
Altogether then, to process the arrival of an element we must temporarily use at most $\cO(k\ell\log(n))$ additional bits of storage, together with the space required to temporarily store this single element.
\end{proof}

\begin{algorithm}[t]
\KwIn{parameters $\ell, k$, independence oracles for $\ell$-matchoid $\cM=\{M_i\}_{i=1}^s$ of rank $k$, weight function $w : X \to \reals$}
\smallskip
\myproc{$\SRepSet$}{
$\outp \gets \emptyset$\;
\ForEach{$e \in X$ arriving in the stream}{
  Let $\outp'$ be the result of running $\RepSet(\outp+e)$\;
  $\outp \gets \outp'$\;
}
\Return $\outp$\;
}
\caption{Streaming FPT-algorithm}
\label{alg:streaming-guess}
\end{algorithm}

\section{Unweighted Coverage Functions in the Value Oracle Model}
\label{sec:an-algor-unwe}

In the previous section, we have focused on the problem of maximizing a \emph{linear} function subject to an $\ell$-matchoid constraint. In this section and the next, we consider the more general \textsc{Maximum $(\cM,z)$-Coverage} problem. Here we are given an $\ell$-matchoid $\cM$,  together with a universe $\pU$ of size $m$, and each element $e \in X$ corresponds to some subset of $\pU$. The goal is then to find a set of elements $S$ that is feasible in $\cM$ and whose union contains at least $z$ points of the universe $\pU$. To avoid confusion, we refer to the elements of $\pU$ as \emph{points} and reserve the term \emph{element} for those elements of $X$ and set variables and functions related to points in $\mathsfit{sans\ serif}$. For each element $e \in X$, we denote by $\points(e)$ the set of points in $\pU$ that corresponds to $e$. Similarly, for any subset $T \subseteq X$, we let $\points(T)$ denote the set of points $\bigcup_{e \in T}\points(e)$ that are covered by at least one element of $T$. In the streaming setting, we suppose that $\pU$ and $X$ are not known in advance, and the elements of $X$ arrive one at a time.


In this section, we consider the case of an \emph{unweighted} coverage function, in which the objective is simply to find a set $S \subseteq X$ of elements that is independent in the given $\ell$-matchoid $\cM$ so that $f(S) = \left|\points(S)\right|$ is maximized. We further suppose that the representation of each element $e$ as a subset $\points(e) \subseteq \pU$ is not directly available, but instead we are given a value oracle for $f$. For any $S \subseteq X$, this oracle returns only the value $f(S)$ (that is, the number of points covered by the union of all elements in $S$). We give a fixed-parameter streaming algorithm constructing a kernel for the problem of finding a feasible set $S$ for an $\ell$-matchoid $\cM$ with $f(S) \geq z$, where $z,\ell \in \posints$ are the parameters. Recall that we can assume that for each $e \in X$ we have that $\{e\}$ feasible for $\cM$.  

\subsection{An intuitive description of our approach}
\label{sec:intu-behind-data}

Due to the limitations of the value oracle model, we require a rather sophisticated data structure to achieve our goal. Here we give some informal discussion and intuition; a formal description will follow. 

Consider any feasible set $O$ for our $\ell$-matchoid $\cM$, with $f(O) \geq z$. Fix some $b_r \in O$. Under what conditions are we justified in throwing away $b_r$ when it arrives in the stream? Here we are primarily concerned with the case in which $b_r$ is critically contributing to the value $f(O)$, so that $f(O) \geq z$ but $f(O - b_r) < z$. Intuitively, even in this case we can throw away $b_r$ if we have stored enough elements to ensure that there exists an element $e$ with the properties that 
\begin{itemize} 
\item[(i)] $O - b_r + e$ is a feasible set in $\cM$;
\item[(ii)] $e$ covers at least as many points outside of $O-b_r$ as $b_r$ itself, i.e., $f(e | O-b_r) \geq f(b_r|O-b_r)$. 
\end{itemize}
To achieve (i) we can simply utilize the joint representative sets introduced in the preceding section. However, guaranteeing (ii) is
trickier. Here, we must ensure that our replacement $e$ covers at least as many points outside of $\points(O-b_r)$ as $b_r$ does and, unlike in the case of linear functions, this marginal coverage will, in general, depend on how both $e$ and $b_r$ interacts with $O - b_r$. One simple approach would be to ensure that we store a representative $e$ for $b_r$ that covers a \emph{superset} of the points covered by $b_r$. However, this may require storing a prohibitively large number of elements: consider the case in which each element that arrives covers some \emph{distinct} set of $t$ points. 

Thus, we adopt a different approach. First, let us do some wishful thinking: imagine that after processing $b_r$, we have $z$ disjoint $z$-representative sets $R_1,\cdots, R_z$ for the set of elements $T$ that have arrived so far, with the following three properties:
\begin{enumerate}
\item[(a)] Each element $e$ in $\cup_{i=1}^{z}R_i$ has the same value $f(e) = f(b_r)$;
\item[(b)]  There exists a set $\pA \subseteq \pU$ of points that are shared by 
 all elements in $\cup_{i=1}^{z}R_i$ and the element $b_r$;
\item[(c)]  No two elements  in $\cup_{i=1}^{z}R_i$ share any point outside $\pA$.
\end{enumerate}
Note that these properties are more relaxed than the requirement that all elements $e$ in our representative set have $\points(b_r) \subseteq \points(e)$: here we require only that $e$ covers some subset $\pA$ of the points in $\points(b_r)$. However, we now further require that there are $z$ distinct such representative sets, and that the stored elements $e$ each cover a disjoint set of points in $\pU \setminus \pA$. 

We now show briefly why this suffices to satisfy property (ii). Given a collection of representative sets $R_1,\ldots,R_z$ satisfying (a)--(c), we can find $z$ distinct representatives (one from each $R_i$) for $b_r$. By the given properties, each of these elements will cover the same set of points $\pA$ as $b_r$, together with $f(b_r) - |\pA|$ unique points outside $\pA$. Then, since $f(O - b_r) < z$, property (b) and the pigeonhole principle imply that for at least one such representative element $e$, the set $\points(e) \setminus \pA$ must be disjoint from $\points(O - b_r)$. This element then covers all the points of $\pA$ that $b_r$ covers, together with a new set of $f(b_r) - |\pA|$ points not covered by any set in $O - b_r$. Thus (as we will formally show) it must satisfy property (ii).

The question now becomes how we can efficiently ensure that some collection of representative sets satisfying the above properties with respect to some set of points $\pA$ exists for any possible $b_r$. To do this, we maintain a tree of such collections for each possible value of $f(b_r) \in \{1,\ldots,z-1\}$. The nodes of each such tree will correspond to some set $\pA$ of commonly covered points, as above, and each node will store a collection of $z$ representative sets satisfying our properties (a)--(c) (with respect to the set $\pA$ of points) for each element that has previously arrived. Note that the algorithm only has access to a value oracle, so we do not know the precise value of the set $\pA$, only that some such common set \emph{exists}. The root node in each tree corresponds to $\pA = \emptyset$. Suppose that $\nn$ is a general node in the tree and that all elements stored in $\nn$ cover the common set $\pA$ of points. Then, for each element $e$ stored in $\nn$, we will potentially create a child node $\nn'$ of $\nn$ associated with $e$. The elements stored in each such child will cover a common set of points $\pA'$ where $\pA \subset \pA' \subseteq \points(e)$. Then, note that as we descend the tree, the set $\pA$ associated with our current node grows larger, and so at depth (at most) $z$, we will have $|\pA| = z$. 

When a new element $b_r$ arrives, we then consider the tree corresponding to $f(b_r)$. We will then descend through this tree until a node corresponding to some set $\pA$ with desired properties (a)--(c) is found. For any tree node, we must determine whether $b_r$ covers some same set of points $\pA$ as all elements in $\cup_{i=1}^{z}R_i$ and $b_r$ does not cover any point outside of $\pA$ that is covered by elements in $\cup_{i=1}^{z}R_i$---we show that this can be accomplished even if we have only a value oracle (and so do not know the points in the underlying set $\pA$). Once a tree node is found whose representative sets $R_1,\ldots,R_z$ satisfy properties (a)--(c) for $b_r$,   we can try to add $b_r$ into exactly one of the $z$-representative sets, $R_1,\cdots, R_z$. If $b_r$ cannot be added to any of these sets, then we must already have a set of $z$ representatives for $b_r$ as described above, and so $b_r$ can be safely thrown away.



\subsection{Formal description of the algorithm}
\label{sec:formal-description}

We now describe our algorithm formally. See the code in Algorithm~\ref{alg:streaming-guess-oracle}. After introducing the required notation, we give a concrete example of how the algorithm behaves for the tree structure shown in Figure~\ref{fig:dynamic}. 
Our algorithm maintains a collection of $z-1$ trees, each storing multiple joint $z$-representative sets. More precisely, for each $1 \leq j \leq z-1$, we maintain a tree that stores only elements $e$ with $f(e) = j$. Each node $\nn$  of our trees will maintain a collection of $z$ disjoint representative sets $R_1,\ldots,R_z$, with the property that each element in these sets covers some common set $\pA$ of at least $d$ points---where $d$ is the depth of $\nn$ in the tree---and no other points in common with any other element. Let $\all(\nn)$ denote the union $\bigcup_{j = 1}^zR_j$ of all the joint $z$-representative sets stored at $\nn$. A node $\nn$ has multiple children associated with each element $e \in \all(\nn)$. For each such child node $\nn'$ of $\nn$, we let $\parent(\nn')$ denote the element $e \in \all(\nn)$ that $\nn'$ is associated with. Specifically, consider some $e \in \all(\nn)$. Then $e$, as well as all other elements of $\all(\nn)$, cover some common set of points $\pA$. For each subset $\pQ$ of points in $\points(e) \setminus \pA$, we will potentially create a new child node $\nn'$ of $\nn$, with $\parent(\nn') = e$. This child node will store only elements that cover precisely this set $\pQ$ of points in $\points(e) \setminus \pA$, together with all points of $\pA$, and cover no other points in common. That is, all elements in this child $\nn'$ thus cover precisely the set of points $\pA \cup \pQ \subseteq \points(e) = \points(\parent(\nn'))$ and the sets of points $\points(e') \setminus (\pA \cup \pQ)$ for all elements $e' \in \all(\nn')$ are disjoint. For convenience, at the root node $\nn$ of each tree, we let $\parent(\nn)$ be a single dummy element $\bot$, with $f(\bot) = 0$ (so $\bot$ covers no points in the underlying representation of $f$).

In the value oracle mode, we do not have access to the actual sets of points covered by any element. We show (in Lemma~\ref{lem:point-queries}) that given elements $a,b,x \in X$, we can determine whether $a$ and $b$ cover the same subset of points in $\points(x)$, and whether they cover no other points in common outside of $\points(x)$ using only queries to the value oracle for $f$. 
When a new element $e$ arrives in the stream, we then first check if $f(e) \geq z$. If this is the case, then we can simply return $\{e\}$ as a kernel for the problem. Otherwise, algorithm will update the tree corresponding to $f(e)$ as follows. First, we find a node $\nn$ in the tree so that both $e$ and each element in $\all(\nn)$ cover some common set of points in $\points(\parent(\nn))$ and have no other points in common (note that this common set of points in $\parent(\nn)$ corresponds to the set $\pA$ in our previous discussion). Starting with $\nn$ as the root node $\nn_{f(e)}$ of the tree corresponding to $f(e)$, we test if this is the case. If not, there must be some element $r \in \all(\nn)$ so that $\points(e) \setminus \points(\parent(\nn))$ and $\points(r) \setminus \points(\parent(\nn))$ are not disjoint. We then descend the tree and recursively test whether our desired property holds for $e$ at each child node $\nn'$ associated with $r$. This is accomplished by the procedure $\FindNode(e,\nn)$, which ultimately either returns a suitable node in which $e$ should be stored or creates a new child of some node in the tree. In Lemma~\ref{lem:find-node}, we show that the node ultimately returned by this procedure indeed satisfies the required analogues of our intuitive properties (a)--(c) from Section~\ref{sec:intu-behind-data}.

Once we have an appropriate such node $\nn$, we iteratively attempt to add $e$ into each joint $z$-representative set $R_j$ stored in this node, stopping as soon as we succeed. This is accomplished by the procedure $\ProcessElem(e,\nn)$. In Lemma~\ref{lem:invariants}, we give several invariants that are maintained by our algorithm as a whole. In particular, all elements stored in the joint $z$-representative sets of any node have the desired properties (with respect to the points they cover), that all the joint $z$-representative sets $R_j$ at a node are disjoint, and that once an element is stored in some set $R_j$ it is never removed from this set later (and so also remains in the associated tree). In order to ensure this last property, we assign elements dummy weights in descending order of arrival. As we will show, this ensures that when rebuilding a representative set after the arrival of some element $e$, the procedure $\RepSet$ (which chooses a maximum weight element in each call to $\Guess$) will never exclude a previously selected element in favor of choosing $e$. At the end of the procedure, we return the union of all the joint $z$-representative sets stored at all the nodes of all the trees.


\begin{algorithm}
\KwIn{Parameters $\ell$,$z$, independence oracles for $\ell$-matchoid $\cM=\{M_i\}_{i = 1}^s$, value oracle for $f : 2^X \to \posreals$.}
\smallskip
\myproc{$\StreamingCoverage$}{
  Let $\bot$ be a dummy element with $f(\bot) = 0$ that covers no points\;
  \For{$1 \leq j \leq z-1$}{
    Let $\nn_j$ be a new root node with $\parent(\nn_j)=\bot$
  }
  \ForEach{$e \in X$ arriving in the stream}{
    \lIf{$f(e) \geq z$}{\Return $\{e\}$}
    Let $\hat{\nn}$ = $\FindNode(e,\nn_{f(e)})$\;
    $\ProcessElem(e,\hat{\nn})$\;
  }
  \Return the set of all elements stored in any node of the trees $\nn_1,\ldots,\nn_{z-1}$\;
}
\smallskip
\myproc{$\FindNode(e,\nn)$}{
  Let $p = \parent(\nn)$\;
  \ForEach{$r \in \all(\nn)$}{
    \If(\myc*[f]{$\points(e) \setminus \points(p)$ and $\points(r) \setminus \points(p)$ not disjoint.}){$f(e|p) \neq f(e|\{r,p\})$\label{li:fn-test1}}
    {
      \ForEach{child node $\nn'$ of element $r$}
      {
   \If{$f(r|r') = f(r|e) = f(r|\{r',e\})$ for all $r' \in \all(\nn')$\label{li:fn-test2}}
        {
          \myc{$e$ covers the same points in $\points(r)$ as each $r' \in \all(\nn')$.}
          \Return $\FindNode(e,\nn')$\label{li:fn-rec}\;
        }
      }
      \myc{$e$ covers a new, distinct set of points in $\points(r)$.}
      Create a new child node $\nn'$, with $\parent(\nn') = r$ and $R_j = \emptyset$ for $1 \leq j \leq z$\;
     \Return $\nn'$\label{li:fn-makechild}\;
    }
  }

\Return $\nn$\label{li:fn-returncurrent}\myc*[r]{$\points(e) \setminus \points(p)$ and $\points(r) \setminus \points(p)$ were disjoint for all $r \in \all(\nn)$}
}
\smallskip
\myproc{$\ProcessElem(e,\nn)$}{
  Let $R_1,\ldots,R_z$ be the sets stored in $\nn$\;
  \For{$1 \leq j \leq z$}{
    Let $w$ assign decreasing weights to points of $R_j+e$ in order of arrival.\;
    Let $R'_j$ be the output of $\RepSet(R_j+e)$ for $\cM$ and $w$, with parameters $k = z$ and $\ell$\;
    Replace $R_j$ by $R'_j$\;
    \lIf(\myc*[f]{$e$ was successfully added}){$e \in R'_j$}
    {\Return}
  }
}
\caption{Streaming FPT-algorithm for unweighted maximum coverage}
\label{alg:streaming-guess-oracle}
\end{algorithm}

\subsection{An example}
\label{sec:an-example}

In order to illustrate the operation of our algorithm and the desired properties of our data structure we now provide a small example. Figure~\ref{fig:dynamic} shows a tree $\nn_3$ storing elements $e$ with $f(e) = 3$ in the case that the parameter $z = 4$. Each node in the tree stores 4 representative sets, illustrated as a rounded box with 4 separate regions. For each child $\nn'$ of a node $\nn$ in a tree, we have illustrated the set $\pQ$ of points in the parent element $\parent(\nn)$ with which this child is associated. Suppose that elements $e_{1},\ldots,e_{14}$ have already arrived and been stored in the tree. We will explain how $e_{15}$ is processed and stored. Because we are working in the value oracle model, the algorithm does not have direct access to the underlying set of points corresponding to each element. However, in order to concretely illustrate our main ideas, here we show this underlying representation. In Lemma~\ref{lem:point-queries}, we argue that all of the tests we perform here can be effectively carried out using only the value oracle for $f$.

\begin{figure}
 \centering 
\begin{tikzpicture}[scale=0.7]
\draw[thick, red, rounded corners=3pt] (0.25, -1.86) ellipse (0.5 and 0.25);
\draw[thick, red, rounded corners=3pt] (0.75, -1.86) ellipse (0.5 and 0.25);
\draw[thick, red, rounded corners=3pt] (0, -1.86) ellipse (0.15 and 0.15);

\draw (2.15, 0.75) node[above] {$\bot$};
\draw[thick] (2.15, 0.75) -- (2.15,0.33);

\coordinate (a) at (-0.15, - 1.86);
\coordinate (b) at (0.25, -0.25 - 1.86);
\coordinate (c) at (1.25, - 1.86);

\treeNodeThree{0}{0}{\pa_1}{\pa_2}{\pa_3}{e_1}{black}{black}{black}
\treeNodeThree{0}{-1.2}{\pb_1}{\pb_2}{\pb_3}{e_2}{black}{black}{black}

\treeNodeThree{3}{0}{\pc_1}{\pc_2}{\pc_3}{e_3}{black}{black}{black}

\draw[thick, rounded corners=5pt] (-1.1, -2.5) -- (-1.1,0.33) -- (6.4,0.33) -- (6.4, -2.5) -- cycle;
\draw (1.8, -2.5) -- (1.8,0.33);
\draw (4.8, -2.5) -- (4.8,0.33);
\draw (5.6, -2.5) -- (5.6,0.33);
\draw (-1.1, -1.05) node[left] {$\nn_3$};

\begin{scope}[shift={(-6,-3.5)}]
\draw[thick, blue] (1, -1.86) ellipse (0.15 and 0.15);
\draw[thick, blue] (0.5, -1.86) ellipse (0.15 and 0.15);
\draw[thick, blue] (1, -0.66) ellipse (0.15 and 0.15);
\coordinate (a1) at (1.15, - 0.66);
\coordinate (b1) at (0.5, -0.15 - 1.86);
\coordinate (c1) at (1, -0.15-1.86);

\treeNodeThree{0}{0}{\pb_1}{\pd_1}{\pd_2}{e_4}{red}{black}{black}
\treeNodeThree{0}{-1.2}{\pb_1}{\pf_1}{\pf_2}{e_5}{red}{black}{black}

\draw[thick, rounded corners=5pt] (-1.1, -2.5) -- (-1.1,0.33) -- (4.2,0.33) -- (4.2, -2.5) -- cycle;
\coordinate (n1) at (2.15, 0.33);
\draw (-1.1, -1.05) node[left] {$\mathbf{a}$};

\draw (1.8, -2.5) -- (1.8,0.33);
\draw (2.6, -2.5) -- (2.6,0.33);
\draw (3.4, -2.5) -- (3.4,0.33);
\end{scope}

\begin{scope}[shift={(1.5,-3.5)}]
\coordinate (n2) at (2.15, 0.33);
\treeNodeThree{0}{0}{\pb_1}{\pb_2}{\pf_1}{e_6}{red}{red}{black}
\treeNodeThree{0}{-1.2}{\pb_1}{\pb_2}{\pf_2}{e_7}{red}{red}{black}
\draw (-1.1, -1.05) node[left] {$\mathbf{b}$};

\draw[thick, rounded corners=5pt] (-1.1, -2.5) -- (-1.1,0.33) -- (4.2,0.33) -- (4.2, -2.5) -- cycle;
\draw (1.8, -2.5) -- (1.8,0.33);
\draw (2.6, -2.5) -- (2.6,0.33);
\draw (3.4, -2.5) -- (3.4,0.33);
\end{scope}

\begin{scope}[shift={(8,-3.5)}]
\treeNodeThree{0}{0}{\pb_2}{\pb_3}{\pd_1}{e_8}{red}{red}{black}
\treeNodeThree{0}{-1.2}{\pb_2}{\pb_3}{\pf_2}{e_9}{red}{red}{black}
\coordinate (n3) at (2.15, 0.33);
\draw (-1.1, -1.05) node[left] {$\mathbf{c}$};

\draw[thick, rounded corners=5pt] (-1.1, -2.5) -- (-1.1,0.33) -- (4.2,0.33) -- (4.2, -2.5) -- cycle;
\draw (1.8, -2.5) -- (1.8,0.33);
\draw (2.6, -2.5) -- (2.6,0.33);
\draw (3.4, -2.5) -- (3.4,0.33);
\end{scope}

\begin{scope}[shift={(-9,-7)}]
\treeNodeThree{0}{0}{\pb_1}{\pf_1}{\pg_1}{e_{10}}{red}{blue}{black}
\treeNodeThree{0}{-1.2}{\pb_1}{\pf_1}{\pg_2}{e_{11}}{red}{blue}{black}
\coordinate (n4) at (2.15, 0.33);

\draw[thick, rounded corners=5pt] (-1.25, -2.5) -- (-1.25,0.33) -- (4.2,0.33) -- (4.2, -2.5) -- cycle;
\draw (1.8, -2.5) -- (1.8,0.33);
\draw (2.6, -2.5) -- (2.6,0.33);
\draw (3.4, -2.5) -- (3.4,0.33);
\draw (1.65, -2.5) node[below] {$\mathbf{d}$};
\end{scope}

\begin{scope}[shift={(-3,-7)}]
\treeNodeThree{0}{0}{\pb_1}{\pf_2}{\pg_2}{e_{12}}{red}{blue}{black}
\coordinate (n5) at (-0.5, 0.33);

\draw[thick, rounded corners=5pt] (-1.25, -2.5) -- (-1.25,0.33) -- (4.2,0.33) -- (4.2, -2.5) -- cycle;
\draw (1.8, -2.5) -- (1.8,0.33);
\draw (2.6, -2.5) -- (2.6,0.33);
\draw (3.4, -2.5) -- (3.4,0.33);
\draw (1.65, -2.5) node[below] {$\mathbf{e}$};
\end{scope}

\begin{scope}[shift={(3,-7)}]
\treeNodeThree{0}{0}{\pb_1}{\pd_2}{\pf_1}{e_{13}}{red}{blue}{black}
\treeNodeThree{0}{-1.2}{\pb_1}{\pd_2}{\pg_1}{e_{14}}{red}{blue}{black}
\treeNodeThree{3}{0}{\pb_1}{\pd_2}{\pg_2}{e_{15}}{red}{blue}{black}
\coordinate (n6) at (-1, 0.33);

\draw[thick, rounded corners=5pt] (-1.25, -2.5) -- (-1.25,0.33) -- (6.4,0.33) -- (6.4, -2.5) -- cycle;
\draw (1.8, -2.5) -- (1.8,0.33);
\draw (4.8, -2.5) -- (4.8,0.33);
\draw (5.6, -2.5) -- (5.6,0.33);
\draw (3, -2.5) node[below] {$\mathbf{f}$};
\end{scope}

\draw[thick,red] (a) -- (n1);
\draw[thick,red] (b) -- (n2);
\draw[thick,red] (c) -- (n3);
\draw[thick,blue] (a1) -- (n6);
\draw[thick,blue] (b1) -- (n4);
\draw[thick,blue] (c1) -- (n5);
\end{tikzpicture}

\caption{Let $z=4$. Here we demonstrate a tree of depth 2 rooted at node $\nn_3$. This tree contains elements $e$ 
 with $f(e)=3$.
 Each node in the tree has $4$ $z$-representative sets.} 
\label{fig:dynamic}
\end{figure}

Since $f(e_{15}) = 3$, we begin at the root $\nn_3$ of the tree shown. At the root $\nn_3$, $\all(\nn_3) = \{e_1,e_2,e_3\}$ and $p = \parent(\nn_3) = \bot$ (recall that $\bot$ is a dummy element with $\points(\bot) = \emptyset$). For each element of $r \in \all(\nn_3) = \{e_1,e_2,e_3\}$, we check in turn whether $\points(e_{15}) \setminus \points(p)$ and $\points(r) \setminus \points(p)$ are disjoint. We find that $\points(e_{15}) \setminus \emptyset = \{\pb_1,\pd_2,\pg_2\}$ and $\points(e_2) \setminus \emptyset = \{\pb_1,\pb_2,\pb_3\}$ are not disjoint. Thus, we will attempt to store $e_{15}$ in some child of $e_2$. 

We consider each child node $\na$, $\nb$, and $\nc$ associated with $e_2$ in turn. To select an appropriate child node, we check whether $e_{15}$ covers the same set of points in $\points(e_2)$ as all of the other elements stored in this child. This is the case for $\na$, since $e_{15}$, $e_4$, and $e_5$ all cover precisely $\pb_1$ and no other element from $\points(e_2)$. Thus, we attempt to insert $e$ into $\na$. At this stage, we have $\nn = \na$, and $p = e_2$. 

We check whether $\points(e_{15}) \setminus \points(p) = \points(e_{15}) \setminus \points(e_2)$ is disjoint from $\points(r) \setminus \points(p) = \points(r) \setminus \points(e_2)$ for each $r \in \all(\na)$. This is not the case, since $e_{15}$ and $e_4$ both cover the point $\pd_2 \not\in \points(e_2)$. Thus, we descend the tree again and consider all child nodes associated with element $e_4$. There is exactly one such child $\nf$. Now, we check if $e_{15}$ covers the same set of points in $e_{5}$ as each element $r \in \all(\nf)$. Indeed, each of these elements (and $e_{15}$) covers precisely the same set of points $\{\pb_1,\pd_2\} \subseteq \points(e_5)$. Thus, we will attempt to insert $e_{15}$ into $\nn = \nf$ with $p = \parent(\nf) = e_5$. 

Now, we find that $\points(e_{15}) \setminus \points(p)$ and $\points(r) \setminus \points(p)$ are disjoint for all $r \in \all(\nf)$. Thus, we process $e_{15}$ at this node. Suppose we try to insert $e_{15}$ into $R_{1} = \{e_{13},e_{14}\}$, but do not succeed. Then, we will try to insert $e_{15}$ into $R_2 = \{\}$ and succeed, giving the tree shown.

Observe that the way in which elements are processed ensures that the elements of $\all(\nn)$ at any node $\nn$ cover precisely the same set of points in $\parent(\nn)$ and, except for these points, are otherwise pairwise disjoint.

\subsection{Analysis}
\label{sec:analysis}
We now give the formal statements of the lemmas discussed above, and carry out our analysis. First we argue that we can determine the required properties of the (unknown) sets of points corresponding to each element by using only a value oracle for the associated coverage function.
\begin{lemma}
\label{lem:point-queries}
For any elements $a,b,x \in X$, 
\begin{enumerate}[topsep=0pt,itemsep=-1ex,partopsep=1ex,parsep=1ex]
\item $\points(a) \cap \points(x) = \points(b) \cap \points(x)$ if and only if $f(x | a) = f(x | b) = f(x | \{a,b\})$.
\item $\points(a) \setminus \points(x)$ and $\points(b) \setminus \points(x)$ are disjoint if and only if $f(a|x) = f(a|\{b,x\})$.
\end{enumerate}
\end{lemma}

\begin{proof}
For part 1, note that
\begin{align*}
f(x | a) &=  |\points(x)| - |\points(x) \cap \points(a)| \\
f(x | b) &=  |\points(x)| - |\points(x) \cap \points(b)| \\
f(x | \{a,b\}) &=  |\points(x)| - |(\points(x) \cap \points(a)) \cup (\points(x) \cap \points(b))|\,.
\end{align*}
If $a$ and $b$ share the same set of points in $x$, then $\points(x) \cap \points(a) = \points(x) \cap \points(b)$ and so all three of the above quantities are equal. On the other hand, if $a$ and $b$ do not share the same set of points in $x$ then there must be some point in only one of $\points(a) \cap \points(x)$ and $\points(b) \cap \points(x)$. Suppose without loss of generality that this point is in $\points(a) \cap \points(x)$. Then, we have $|(\points(a) \cap \points(x)) \cup (\points(b) \cap \points(x))| > |\points(b) \cap \points(x)|$ and so the above 3 quantities are \emph{not} equal.

For part 2, note that 
\begin{align*}
f(a | x) &= |\points(a)| - |\points(a) \cap \points(b) \cap \points(x)| 
- |\points(a) \cap (\points(x) \setminus \points(b))| \\
f(a | \{b,x\}) &= |\points(a)| - |\points(a) \cap \points(b) \cap \points(x)| - |\points(a) \cap (\points(x) \setminus \points(b))| - |\points(a) \cap (\points(b) \setminus \points(x))|\,.
\end{align*}
The above two quantities are equal if and only if $|\points(a) \cap (\points(b) \setminus \points(x))| = 0$ and so $a$ and $b$ do not share any points other than those in $x$.
\end{proof}

We now show that our procedure $\FindNode$ returns a node $\nn$ of the appropriate tree satisfying the formal analogues of the intuitive properties (a)--(c) described in Section~\ref{sec:intu-behind-data}.
\begin{lemma}
\label{lem:find-node}
Let $\hat{\nn}$ be the node returned by $\FindNode(e,\nn_{f(e)})$. Let $d \geq 0$ be the depth of $\hat{\nn}$ and $\hat{p} = \parent(\hat{\nn})$. Then,
\begin{enumerate}[topsep=0pt,itemsep=-1ex,partopsep=1ex,parsep=1ex]
\item $|\points(e) \cap \points(\hat{p})| \geq d$ and $\points(e) \cap \points (\hat{p}) = \points(r) \cap \points(\hat{p})$ for all $r \in \all(\nn)$. 
\item $\points(e) \setminus \points(\hat{p})$ and $\points(r) \setminus \points(\hat{p})$ are disjoint for all $r \in \all(\nn)$.
\end{enumerate}
\end{lemma}
\begin{proof}
Consider any recursive call made to $\FindNode(e,\nn)$ during the execution $\FindNode(e,\nn_{f(e)})$ and let $p = \parent(\nn)$. We claim that in any such call $|\points(e) \cap \points(p)| \geq d$ and $\points(e) \cap \points (p) = \points(r) \cap \points(p)$ for all $r \in \all(\nn)$. This is true for $\nn=\nn_{f(e)}$, as $p = \parent(\nn_{f(e)}) = \bot$ and $\points(\bot) = \emptyset$. Suppose that the claim holds for all nodes of depth at most $d$, and let $\nn$ be a node of depth $d+1$ such that executing $\FindNode(e,\nn_{f(e)})$  results in a call to $\FindNode(e,\nn')$. This call was made in line~\ref{li:fn-rec} when executing some immediate predecessor call $\FindNode(e,\nn)$ in the recursion, where the depth of $\nn$ was $d$. Let $r = \parent(\nn')$ and $p = \parent(\nn)$. Then, $r \in \all(\nn)$ and by the induction hypothesis, $|\points(e) \cap \points(p)| \geq d$ and $\points(e) \cap \points(p) = \points(r) \cap \points(p)$. Additionally, due to  line~\ref{li:fn-test1} we must have $f(e|p) \neq f(e|\{r,p\})$ and so by Lemma~\ref{lem:point-queries}(2), $\points(e) \setminus \points(p)$ and $\points(r) \setminus \points(p)$ share at least one point. Thus, $|\points(e) \cap \points(r)| \geq d + 1$, as required. Due to line~\ref{li:fn-test2}, we must also have $f(r|r') = f(r|e) = f(r|\{r',e\})$ for all $r' \in \all(\nn')$, and so by Lemma~\ref{lem:point-queries}(1), $\points(e) \cap \points(r) =\points(r') \cap \points(r)$ for all $r' \in \all(\nn')$, as required. This completes the proof of the induction step.

Now to prove the lemma, we note that any node $\hat{\nn}$ returned by $\FindNode(e,\nn_{f(e)})$ must either be returned directly by some call $\FindNode(e,\hat{\nn})$ in line~\ref{li:fn-returncurrent} or be a new child node of some $\nn$,  returned by $\FindNode(e,\nn)$ in line~\ref{li:fn-makechild}. If the former case happens, then the first part of the lemma follows immediately from the claim above. Moreover, in this case, we must have $f(e|\hat{p}) = f(e|\{r,\hat{p}\})$ for all $r \in \all(\hat{\nn})$ due to line~\ref{li:fn-test1} and so by Lemma~\ref{lem:point-queries}(2), $\points(e) \setminus \points(\hat{p})$ and $\points(r) \setminus \points(\hat{p})$ are disjoint for all $r \in \all(\hat{\nn})$.

On the other hand, suppose that $\hat{\nn}$ is returned as a new child of some node $\nn$, and let $p = \parent(\nn)$ and $d$ be the depth of $\nn$. Then, due to line~\ref{li:fn-test1}, we must have $f(e|p) \neq f(e|\{r,p\})$ for $r = \parent(\hat{\nn}) = \hat{p}$ and so, by Lemma~\ref{lem:point-queries}(2), there is at least one point in both $\points(e) \setminus \points(p)$ and $\points(r) \setminus \points(p)$. By our induction claim above, $|\points(e) \cap \points(p)| \geq d$ and $\points(e) \cap \points(p) = \points(r) \cap \points(p)$. Thus, $|\points(e) \cap \points(\hat{p})| = |\points(e) \cap \points(r)| \geq d+1$, which is the depth of $\hat{\nn}$. Moreover, $\all(\hat{\nn}) = \emptyset$, so the rest of the lemma holds trivially.
\end{proof}

Let $\nn$ be the node returned by $\FindNode(e,\nn_{f(e)})$ when $e$ arrives. 
The next lemma summarizes the property guaranteed by $\nn$. The first two items are easy consequences of 
Lemma~\ref{lem:find-node}, while the last two items are the invariants guaranteed by the way in which we update representative sets in 
$\ProcessElem(e,\nn)$.

\begin{lemma}
\label{lem:invariants}
Suppose that $e$ is processed at a node $\nn$ by $\ProcessElem(e,\nn)$ during Algorithm~\ref{alg:streaming-guess-oracle}, and let $p = \parent(\nn)$. Then,
\begin{enumerate}[topsep=0pt,itemsep=-1ex,partopsep=1ex,parsep=1ex]
\item $\points(e) \cap \points(p) = \points(r) \cap \points(p)$ for all $r \in \all(\nn)$ and $|\points(e) \cap \points(p)|$ is at least the depth of $\nn$.
\item $\points(e) \setminus \points(p)$ is disjoint from $\points(r) \setminus \points(p)$ for all $r \in \all(\nn)$.
\item $e$ is added to at most one set $R_j$ stored in $\nn$. Thus, all sets $R_j$ stored in $\nn$ are mutually disjoint.
\item
If $e$ is added to some $R_j$ stored in $\nn$ then $e$ stays always as part of $R_j$ even after any subsequent call to $\ProcessElem(x,\nn)$ for $x \in X$.
\end{enumerate}
\end{lemma}

\begin{proof}
Parts (1) and (2) follow directly from Lemma~\ref{lem:find-node}. Part (3) follows immediately from the fact that we return from the loop in $\ProcessElem(e,\nn)$ the first time that $e$ is added to one of the sets $R_1,\ldots,R_z$, and so $e$ is added to at most one set. Intuitively, part (4) follows from the fact that we assign elements dummy weights in descending order of arrival in line 25 of the algorithm. Thus, the procedure $\RepSet$ will always prefer adding an element already in the representative set $R_j$ over the new element $x$.

Formally, consider some set $R_j$ stored at node $\nn$ at some time during the algorithm and consider any element $e \in R_j$. Let $T$ be the set of elements for which $R_j$ was the output of $\RepSet(T)$. Consider the next element $x$ for which $\ProcessElem(x,\nn)$ is called to update $R_j$, and let $R_j'$ be the resulting set output by $\RepSet(R_j + x)$. We will show that $e \in R_j'$ as well. Part (4) then follows by induction on the stream of elements processed at node $\nn$. 

To prove that $e \in R'_j$ as claimed, we show by induction on $\|J\|$ that if there is some call
$\Guess(J,Y)$ in $\cT$, then there must also be a corresponding call
$\Guess(J,A)$ in $\cT'$ with
$(Y\cap R_j) \subseteq A \subseteq (Y\cap R_j)+x$. 
Then, since the elements of $R_j$ are precisely those that are added to $J$ by some call $\Guess(J,Y)$ in $\cT$, we must have $e \in J$ for some call $\Guess(J,Y)$ in $\cT$. The corresponding call $\Guess(J,A)$ in $\cT'$ will then also have $e \in J$ and so $e$ will be included in $R'_j$.

To prove the inductive claim, first we note that $\|J\| = 0$,
the roots of $\cT$ and $\cT'$ correspond to calls
$\Guess((\emptyset,\ldots,\emptyset),T)$ and
$\Guess((\emptyset,\ldots,\emptyset),R_j+x)$, respectively, and so the
claim follows, as $R_j \subseteq T$. Suppose now that the claim holds for all 
$\|J\| \leq d$ and consider some call $\Guess(J,Y)$ in $\cT$,
where $\|J\| = d+1 > 0$. Consider the parent $\Guess(J',Y')$ of this
call in $\cT$ so that $Y = Y' \setminus \spa_{M_i}(J'_i + e')$ and
$J = J' +_i e'$ for some $e' = \argmax_{a \in Y'}w(a)$ with
$X_i \in X(e')$, and $\|J'\| \leq d$. By the induction hypothesis,
there is then some corresponding call $\Guess(J',A')$ in $\cT'$, where
$(Y' \cap R_j) \subseteq A' \subseteq (Y' \cap R_j) +x$. Since $e'$
was selected by $\Guess(J',Y')$, $e' \in Y' \cap R_j$ and so
$e' = \arg\max_{a \in Y'}w(a) = \arg\max_{a \in Y' \cap
  R_j} w(a)$. Moreover, since $e' \in R_j$, $x$ must arrive after $e'$, and so $w(x) < w(e')$ and
$Y' \cap R_j \subseteq A' \subseteq (Y' \cap R_j) + x$. Thus 
$e' = \arg\max_{a \in A'}w(a)$ and so $\Guess(J',A')$ will also select
$e'$, resulting in a child call $\Guess(J,A)$, with $J = J' +_i e'$
and $A = A' \setminus \spa_{M_i}(J'_i + e')$. To complete the
induction step, it remains show that
$(Y \cap R_j) \subseteq A \subseteq (Y \cap R_j) + x$. For any
$y \in (Y \cap R_j)$ we must have $y \in (Y' \cap R_j) \subseteq A'$
and also $y \not\in \spa_{M_i}(J'_i + e')$, since
$y \in Y = Y'\setminus \spa_{M_i}(J'_i + e')$. Then, $y \in A$ as
well. Similarly, for any $a \in A - x$ we must have
$a \in A' - x \subseteq (Y'\cap R_j)$ and
$a \not\in \spa_{M_i}(J'_i+e')$, since
$A = A' \setminus \spa_{M_i}(J'_i + e')$. Then, $a \in (Y \cap R_j)$ 
as well. Thus, $(Y \cap R_j) \subseteq A \subseteq (Y\cap R_j)+x$, as
required. This completes the induction step.
\end{proof}

Using the above properties, we now prove our main result. Note that once the set $R$ has been computed, the following theorem implies that we can find a set $S$ that is feasible for $\cM$ with $f(S) \geq z$ by using at most $|R|^z$ value oracle queries, as in the proof we will show that our kernel $R$ 
must in fact contain a feasible set of \emph{at most $z$} elements $S$ with $f(S) \geq z$. It follows that, when parameterized by $z =  f(O)$, we can optimize an unweighted coverage function $f$ in a general $\ell$-matchoid using at most a polynomial number of value queries to $f$.
\begin{theorem}
\label{thm:main-oracle}
Consider an $\ell$-matchoid $\cM=\{M_i=(X_i, \I_i)\}_{i=1}^s$ and let $f : 2^X \to \posints$ be a value oracle for an unweighted coverage function. Then Algorithm~\ref{alg:streaming-guess} produces a kernel $R$ for finding a feasible set $S$ in $\cM$ with $f(S) \geq z$. For $\ell = 1$, $|R| \leq  N_{1,z} \triangleq \cO\bigl(2^{(z-1)^2}z^{2z+1}\bigr)$ and the algorithm requires storing $\cO(N_{1,z}\log n)$ bits in total and makes at most $\cO(\Gamma_{1,z} 2^{z-1}z^2 n)$ value queries to $f$, where $\Gamma_{1,z}=z$. 
For $\ell > 1$, $|R| \leq N_{\ell,z} \triangleq \cO\bigr(2^{(z-1)^2}\ell^{z(z-1)\ell}z^{z+1}\bigl)$ and the algorithm requires storing at most $\cO(N_{\ell,z}\log n)$ bits in total and makes at most $\cO(\Gamma_{l,z}2^{z-1}z^2 n)$ value queries to $f$, 
where $\Gamma_{l,z} = \cO(l^{(z-1)l})$. 

\end{theorem}
\begin{proof}
Let $O = \{b_1,\ldots,b_{k}\}$ with $f(O) \geq z$. We may suppose without loss of generality that $k \leq z$, as follows. Fix some set $\pZ$ of $z$ points covered by $O$. 
Then, as long as $|O| \geq z$, there must exist some $b_r \in O$ that can be removed from $O$ while leaving all of $\pZ$ covered. Then, $O - b_r$ is feasible for $\cM$ with $f(O - b_r) \geq z$.

Let $R$ be the output of Algorithm~\ref{alg:streaming-guess-oracle}. If $f(b_{r}) \geq z$, for some $b_r$, then Algorithm~\ref{alg:streaming-guess-oracle} will return an element $e$ so that $f(e) \geq z$ and the theorem holds trivially. 
So in the following, we suppose that $f(b_r) \leq z-1$ for all $1 \leq r \leq k$ and so every $b_r \in O$ is processed by some call to $\ProcessElem(b_r,\nn)$ upon arrival, where $\nn$ is the node returned by $\FindNode(b_r,\nn_{f(b_r)})$.

We show by induction on $0 \leq r \leq k$ that there is some set $S_r \subseteq R$ with $O_r = O \setminus \{b_1,\ldots,b_r\} \cup S_r$ feasible for $\cM$, $|S_r| \leq r$, and $f(O_r) \geq z$. For $r = 0$, this holds trivially by setting $S_0 = \emptyset$ and $O_0 = O$. For the general case $r>0$, 
We show how to construct $O_{r}$, assuming that $O_{r-1}$ with the desired properties exists. 
If $f(O_{r-1} - b_{r}) \geq z$, then letting $S_{r} = S_{r-1}$ we have $f(O_{r}) = f(O_{r-1} - b_{r}) \geq z$ and so the claim holds easily. Thus in the following we assume that $f(O_{r-1} - b_{r}) \leq z-1$.

Let $\nn$ be the node that processes element $b_{r}$ (i.e.\ the node for which $\ProcessElem(b_{r},\nn)$ is called in Algorithm~\ref{alg:streaming-guess-oracle}) and let $p = \parent(\nn)$. If $b_{r}$ is added to some set stored in $\nn$, then letting $S_{r} = S_{r-1} + b_{r} \subseteq R$ we have $|S_{r}| = |S_{r-1}| +1 \leq r$ and $f(O_{r}) = f(O_{r-1}) \geq z$ as required. On the other hand if $b_{r}$ is not added to any of the sets $R_1,\ldots,R_z$ stored in $\nn$, then Algorithm~\ref{alg:streaming-guess-oracle} must have called $\RepSet(R_j + b_{r})$ to construct a joint $z$-representative set for $(R_j + b_{r},\cM,w)$ for each $1 \leq j \leq z$. 
By Theorem~\ref{thm:main-repset} there is then some element $e_j$ in the output of each of these sets so that $O_{r-1} - b_{r} +e_{j}$ is feasible in $\cM$. By Lemma~\ref{lem:invariants}(1), (2), and (3), each $e_j$ has $\points(e_j) \cap \points(p) = \points(b_{r}) \cap \points(p)$ and the sets $\{\points(e_j) \setminus \points(p)\}_{j = 1}^z$ are mutually disjoint. As $|\points(O_{r-1} - b_{r})| = f(O_{r-1}-b_{r}) \leq z-1$, we have at least one element $e \in \{e_1,\ldots,e_z\}$ for which $\points(e) \setminus \points(p)$ is disjoint from $\points(O_{r-1} - b_{r})$.
Let $S_{r} = S_{r-1} + e$.  Then $|S_{r}| = |S_{r-1}| + 1 \leq r$ as required. Additionally, $O_{r} = O_{r-1} - b_{r} + e$ is feasible for $\cM$. By Lemma~\ref{lem:invariants}(4), we will also have $e \in R$ for the final set of elements $R$ produced by Algorithm~\ref{alg:streaming-guess-oracle}.

It remains to show that $f(O_{r}) \geq z$. We note that:
\begin{align*}
f(e|O_{r-1}-b_{r}) &= |\points(e)| - |\points(e) \cap \points(O_{r-1} -b_{r})|  \\
&= |\points(e)| - |(\points(e) \cap \points(p)) \cap \points(O_{r-1}-b_{r})|
- |(\points(e) \setminus \points(p)) \cap \points(O_{r-1}-b_{r})| \\
&= |\points(e)| - |(\points(e) \cap \points(p)) \cap \points(O_{r-1}-b_{r})| \\
&= |\points(b_{r})| - |(\points(b_{r}) \cap \points(p)) \cap \points(O_{r-1}-b_{r})| \\
&\geq |\points(b_{r})| - |\points(b_{r}) \cap \points(O_{r-1}-b_{r})| \\
&= f(b_{r} | O_{r-1}-b_{r})\,,
\end{align*}
where the third equation follows from the fact that $\points(e) \setminus \points(p)$ is disjoint from $\points(O_{r-1} - b_{r})$, and the fourth from $\points(e) \cap \points(p) = \points(b_{r}) \cap \points(p)$, as well as $|\points(b_{r})| = f(b_{r}) = f(e) = |\points(e)|$ since both $e$ and $b_{r}$ were stored in the same tree. Thus:
\[
f(O_{r}) = f(e|O_{r-1} - b_{r}) + f(O_{r-1} - b_{r})\geq 
f(b_{r} | O_{r-1} - b_{r}) + f(O_{r-1} - b_{r}) = f(O_{r-1}) \geq z,
\]
as required. This completes the proof of the induction step. The first claim of the theorem then follows by setting $r = k$ and noting that $O_k  = S_k \subseteq R$ and $|S_k| \leq k \leq z$. 

We next discuss the space requirement. 
Consider some tree with root $\nn_x$, where $1 \leq x \leq z-1$. We note that by Theorem~\ref{thm:matchoid-offline}, each of the $z$ sets $R_1,\ldots,R_z$ stored at any node of a tree has size at most $\Gamma_{\ell,z} \triangleq \sum_{q = 0}^{(z-1)\ell}\ell^q$. Thus each node in the tree stores at most $\Gamma_{\ell,z}z$ elements. Whenever $\FindNode(e,\nn)$ returns a new child $\nn'$ with $\parent(\nn') = r \in \all(\nn)$, we must have $\points(e) \cap \points(r) \neq \points(e) \cap \points(r')$ for all $r'$ stored in the child nodes of $\nn$ associated with $r$. By Lemma~\ref{lem:invariants}(1), all the elements $r'$ in each such existing child node cover some common set points covered by $\points(r)$. Thus, any element $r$ can have at most $2^{|\points(r)|} = 2^{f(r)} \leq 2^{z-1}$ associated child nodes. Altogether, then, a node in the tree has at most $2^{z-1}\Gamma_{\ell,z}z$ children. Additionally, Lemma~\ref{lem:invariants}(1) implies that the tree has depth at most $x \leq z-1$. Thus, the total number of nodes in the tree is at most
$N \triangleq \sum_{d = 0}^{z-1}(2^{z-1}\Gamma_{\ell,z}z)^d = \cO\bigl( 2^{(z-1)^2}(\Gamma_{\ell,z})^{z-1}z^{z-1} \bigr)$, with each storing at most $\Gamma_{\ell,z}z$ elements. We maintain $z-1$ such trees, so the total number of elements stored across all trees is at most $\cO\bigl(N\Gamma_{\ell,z} z^2) = \cO\bigl(2^{(z-1)^2}(\Gamma_{\ell,z})^zz^{z+1}\bigr)$. 
The total memory required by the algorithm is at most that required to maintain all of the stored elements in a dynamic tree, which requires $\cO(N\Gamma_{\ell,z} z\log n)$ total bits per tree and so $\cO(N\Gamma_{\ell,z} z^2\log n)$ bits in total. When an element $e$ arrives, we make several calls to $\RepSet(\outp_i+e)$. As shown in the proof of Theorem~\ref{thm:main-streaming}, this can be accomplished by temporarily storing only $\cO(z\ell \log n) = \cO(N \log n)$ further bits. Altogether then, the algorithm stores at most $\cO(N\Gamma_{\ell,z} z^2\log n)$ bits at all times during its execution.
For $\ell = 1$, $\Gamma_{\ell,z} = z$, and so $N\Gamma_{\ell,z} z^2 = \cO\bigl( 2^{(z-1)^2}z^{2z+1}\bigr)$. For $\ell > 1$, $\Gamma_{\ell,z} = \cO(\ell^{(z-1)\ell})$ and so $N\Gamma_{\ell,z} z^2 = \cO\bigl(2^{(z-1)^2}(\Gamma_{\ell,z})^zz^{z+1}\bigr) =  \cO\bigl(2^{(z-1)^2}\ell^{z(z-1)\ell}z^{z+1}\bigr)$. 

Finally, we consider the number of value queries.  When descending the tree in the procedure $\FindNode(e,\nn)$, inside each node $\nn$, 
we need to check possibly all elements stored in $\nn$ in Line 13 using value queries, and there can be at most $\Gamma_{\ell,z}z$ of these. 
Furthermore, if the condition in Line 13 holds for some element $r$, we need to check all its child nodes. There can be $2^{z-1}$ such child nodes, 
and for each one, in Line 15, we need to check all its elements using value queries. Each such child again has at most $\Gamma_{\ell,z}z$ elements. 
In summary, inside each node we need $O(\Gamma_{l,z} 2^{z-1}z)$ oracle calls, implying that a total of 
$O(\Gamma_{l,z}  2^{z-1}z^2)$ value queries for processing one new element.

\end{proof}

\section{Improved Algorithms in the Explicit Model}
\label{sec:impr-algor-expl} 

We now consider the weighted version of \textsc{Maximum $(\cM,z)$-Coverage}, in which we additionally have a weight function $\pw : \pU \to \posreals$ and now must find a feasible set $S$ for $\cM$ that covers (up to) $z$ points of maximum total weight. Unlike Section~\ref{sec:an-algor-unwe}, here a critical difference is that we assume that the sets of points corresponding to each element are given explicitly. 
Note that here we allow our solution to cover more than $z$ points, but consider only the $z$ heaviest points in computing the objective. Thus, if $O$ is some optimal solution, by setting $z$ to be the total number of points covered by $O$, then our results imply that we can find  a solution $S$ that has $f(S) \geq f(O)$. In fact, our result implies a stronger guarantee, as it ensures that the heaviest $z$ points covered by $S$ have alone total weight at  least as large as those covered by $O$.

We combine multiple joint $z$-representative sets with a color coding procedure to obtain our results for Maximum $(\cM,z)$-coverage. To this end, we consider a hash function $h : \pU \to [\bar{z}]$, where $\bar{z}$ is the smallest power of $2$ that is at least $z$. For each point $\pp$, we call the value $h(\pp) \in [\bar{z}]$ the \emph{color} assigned to $\pp$. For an element $e \in X$, we further define $h(e) = \{h(\pp) : \pp \in e\}$ to be the set of all colors that are assigned to the points covered by $e$. 

For any set of points $\pT \subseteq \pU$, we let $\pw(\pT) = \sum_{\pp \in \pT}\pw(\pp)$ denote the total weight assigned to these points by the weight function $\pw$. We fix any solution $O$ to the problem and let $\pZ$ be a set of up to $z$ points covered by $O$. Fix $h : \pU \to [\bar{z}]$. We say that $\pZ$ is \emph{well-colored} by $h$ if $h$ is injective on $\pZ$ (i.e.\ $h$ assigns each point $\pp \in \pZ$ a unique color in $[\bar{z}]$). Suppose now that $\pZ$ is well-colored by $h$. For each possible subset $C \subseteq [\bar{z}]$ of colors, and each set $S \subseteq X$ of elements such that $C \subseteq h(\points(S))$, we define
\[
f_C(S) = \sum_{c \in C}\max\{ \pw(\pp) : \pp \in \points(S),\, h(\pp) = c \}
\]
to be the sum of the weights of the single heaviest point of each color in $C$ that is covered by some element of $S$. Let $X_C = \{e \in X : C \subseteq h(e) \}$ be the set of all elements containing at least one point assigned each color of $C$. Then, we can define $w_C(e) : X_C \to \reals$ by $w_C(e) = f_C(\{e\})$. Note that to compute $w_C(e)$ it is enough to simply remember the ``heaviest'' point in $\points(e)$ of each color. Thus, in the streaming setting we can maintain all of our constructions by using only the set of $|h(e)| \leq \bar{z} = \cO(z)$ points 
and all other points can be discarded.

Let $C$ and $C'$ be two disjoint sets of colors and suppose $A \subseteq X$ with $C \subseteq h(\points(A))$ and $b \in X_{C'}$ (so $C' \subseteq h(\points(b))$). Then, $C\cup C' \subseteq h(\points(A+b))$ and so
\begin{align}
f_C(A) + w_{C'}(b) 
&= \sum_{c \in C}\max \{\pw(\pp) : \pp \in \points(A),\, h(\pp) = c\}
+ \sum_{c \in C'}\max \{\pw(\pp) : \pp \in \points(b),\, h(\pp) = c\}
\nonumber \\
&\leq 
\sum_{c \in C \cup C'} \max \{\pw(\pp) : \pp \in \points(A+b),\, h(\pp) = c\} 
= f_{C \cup C'}(A + b)\,, \label{eq:cover-weight-ineq}
\end{align}
since for any point $\pp \in \points(A)$ with $h(\pp) = c \in C$ or any point $\pp \in \points(b)$ with $h(\pp) = c \in C'$, we must also have $\pp \in \points(A+b)$ with $h(\pp) = c \in C \cup C'$.

Using the above constructions, we now show how to combine multiple $z$-representative sets to obtain a kernel for \textsc{Maximum $(\cM,z)$-Coverage}. 
\begin{lemma}\label{lem:coverage-repset}
Let $\cM = \{M_i=(X_i,\I_i)\}_{i=1}^s$ be an $\ell$-matchoid. Suppose that $R = \bigcup_{C \subseteq [\bar{z}]}R_C$, where each $R_C$ is a joint $z$-representative set for $(X_C,\cM, w_C)$. Let $O$ be any set that is 
feasible for $\cM$. Let $\pZ$ be a set of up to $z$ points of maximum weight covered by $O$ and suppose that $\pZ$ is well-colored by $h$. Then, there is some $S \subseteq R$ that is feasible for $\cM$ and covers $|\pZ|$ points of total weight at least as large as that of $\pZ$.
\end{lemma}
\begin{proof}
Suppose that $O = \{b_1,\ldots,b_{k'}\}$ and fix some set $\pZ$ of $z$ points covered by $O$ that is well-colored by $h$. For each $1 \leq r \leq k'$, let $\pP_r = \bigcup^r_{j = 1}(\pZ \cap \points(b_j))$ be the set of points from $\pZ$ covered by the first $r$ elements of $O$ according to our indexing and let $C_r = h(\pP_r)$ be the set of colors assigned to these points. 
Note that the sets $\{\pP_r \setminus \pP_{r-1}\}_{r=1}^{k'}$ form a partition of $\pZ$. We can suppose without loss of generality that $k' \leq z$ as otherwise there must be some element $b_r \in O$ with $\pP_r \setminus \pP_{r-1} = \emptyset$. Any such element can be removed from $O$ to obtain a feasible solution that still covers all of $\pZ$.

We now show by induction on $0 \leq r \leq k'$, that there exists a set of elements $S_r \subseteq R$ such that $O_r = O \setminus \{b_1,\ldots,b_r\} \cup S_r$ is feasible for $\cM$, $|O_r| = k'$, $C_r \subseteq h(\points(S_r))$, and $f_{C_r}(S_r) \geq \pw(\pP_r)$. In the case that $r = 0$, this follows trivially by letting $S_0 = \emptyset$ and $O_0 = O$.

In the general case $r > 0$, the induction hypothesis implies that there is a set $S_{r-1}$ such that $O_{r-1} = O \setminus \{b_1,\ldots,b_{r-1}\} \cup S_{r-1}$ is feasible for $\cM$, $|O_{r-1}| = k'$, $C_{r-1} \subseteq h(\points(S_{r-1}))$, and $f_{C_{r-1}}(S_{r-1}) \geq \pw(\pP_{r-1})$. 
We will consult the representative set $\outp_C$ associated with the set of colors $C = C_r \setminus C_{r-1}$. Note that $b_r \in X_C$ and so by Theorem~\ref{thm:main-repset}, there must exist some element $e_r$ in $\outp_C$ such that $O_{r-1} - b_r + e_r$ 
is feasible, and $w_{C}(e_r) \geq w_{C}(b_r)$. 
Let $S_r = \{e_1,\ldots,e_r\}$ so that $O_r = O \setminus \{b_1,\ldots,b_r\} \cup S_r = O_{r-1} - b_r + e_r$. Then 
$O_r$ is feasible for $\cM$ and $|O_{r}| = |O_{r-1}| = k'$, as required. Since $e_r \in X_C$, $e_r$ contains a point assigned each color in $C = C_r \setminus C_{r-1}$. Thus, $C_r \subseteq h(\points(S_{r-1} + e_r)) = h(\points(S_r))$, as required.
\begin{multline*}
\pw(\pP_{r}) = \pw(\pP_{r-1}) + \pw(\pP_r \setminus \pP_{r-1})
\leq \pw(\pP_{r-1}) + w_C(b_r) \\
\leq f_{C_{r-1}}(S_{r-1}) + w_C(b_r) 
\leq f_{C_{r-1}}(S_{r-1}) + w_{C}(e_r)
\leq f_{C_{r}}(S_{r})\,,
\end{multline*}
where the first inequality follows since $\pP_r \setminus \pP_{r-1}$ is a subset of points of color $C$ covered by $b_r$, the second from the induction hypothesis, the third from $w_C(e_r) \geq w_C(b_r)$, and the last from \eqref{eq:cover-weight-ineq}. This completes the induction step.

To complete the proof of the lemma, we set $r = k'$. Then, we have $O_{k'} = S_{k'} \subseteq R$ and $|S_{k'}| = |O_{k'}| = k'$. Moreover, by definition $C_{k'} = h(\pp_{k'}) = \pZ$. We have $S_{k'}$ feasible for $\cM$, and $C_{k'} \subseteq h(\points(S_k))$ so $S_{k'}$ contains a distinct point of $\pU$ colored with each color $c \in C_{k'}$. For each color $c$, consider the heaviest such point. The total weight of these points is precisely $f_{C_{k'}}(S_{k'}) \geq \pw(P_{k'}) = \pw(\pZ)$.
\end{proof}

We now give a streaming algorithm for computing the collection of joint $z$-representative sets required by Lemmas~\ref{lem:coverage-repset}. Our procedure $\textsc{StreamingMaxCoverage}$ is shown in Algorithm~\ref{alg:coverage-streaming}. Thus far, we have supposed that some set of points $\pZ$ in a solution $O$ was well-colored by a given $h : \pU \to [\bar{z}]$. Under this assumption, our procedure simply runs a parallel instance of the streaming procedure $\SRepSet$ for each $C \subseteq [\bar{z}]$. When a new element $e \in X$ arrives, we assign each point $\pp \in \points(e)$ a color $h(\pp)$. In order to limit the memory required by our algorithm, we will remember only the maximum-weight point of each color in $e$.
As noted in the proof of Lemmas~\ref{lem:coverage-repset}, the resulting collection of elements $R$ that we produce will still contain a feasible solution $S$ with the necessary properties. It is clear that if we later consider the corresponding set of all points in $\pP(e)$, we can only cover \emph{more} points of $\pU$.

\begin{algorithm}[t]
\KwIn{Parameters $\ell,z$, independence oracles for $\ell$-matchoid $\cM=\{M_i\}_{i=1}^s$, weight function $\pw : \pU \to \reals$, hash function $h : \pU \to [\bar{z}]$.}
\smallskip
\myproc{\textsc{StreamingMaxCoverage}}{
\ForEach{$C \subseteq [\bar{z}]$}{
  Let $\SRepSet_C$ be an instance of the procedure $\SRepSet$ for $\cM$, with output $\outp_C$\;
}
\ForEach{$e \in X$ arriving in the stream}
{
  Color the points $\points(e)$ using $h$\;
  Discard all points from $e$ except for the maximum weight point of each color\;
  \ForEach{$C \subseteq h(\points(e))$}{
    Define $w_C(e) = \sum_{\pp \in \points(e)\, :\, h(\pp) \in C}w(\pp)$\;
    Process the arrival of $e$ in $\SRepSet_C$ with weight $w_C(e)$\;
  }   
}
\Return $R = \bigcup_{C \subseteq [\bar{z}]} \outp_C$\;
}
\caption{Streaming FPT-algorithm for the \textsc{Maximum $(\cM,z)$-Coverage}}
\label{alg:coverage-streaming}
\end{algorithm}

We now consider the problem of ensuring that the given set $\pZ$ of up to $z$ points in $\pP(O)$ is well-colored by $h$. In the offline setting, letting $h$ assign each point $\pp \in \pU$ a color uniformly at random guarantees that this will happen with probability depending on $z$. In the streaming setting, however, we cannot afford to store a color for each point of $\pU$, but must still ensure that a point receives a consistent color in each set that it appears in. To accomplish our goal, we use a \emph{$z$-wise independent family} $\cH$ of hash functions $h : X \to [\bar{z}]$. Such a family $\cH$ has the property that for every set of at most $z$ distinct elements $(e_1,\ldots,e_z) \in X^z$, and any $z$ (not necessarily distinct) values $(c_1,\ldots,c_z) \in [\bar{z}]^z$, the probability that $h(e_1)=c_1, h(e_2)=c_2, \ldots,$ and $h(e_z)=c_z$ is precisely $\bar{z}^{-z}$. A classical result of Wegman and Carter~\cite{Wegman1981} provides a construction of such a family $\cH$ of functions $h : [\bar{m}] \to [\bar{z}]$ when both $\bar{m}$ and $\bar{z}$ are prime powers and storing and computing each function requires a random seed of only $\cO(z \log\bar{m})$ bits. In our setting, it suffices to set $\bar{m}$ to the smallest power of 2 larger than $|\pU|$ to obtain a family of functions $h : \pU \to [\bar{z}]$, each of which can be stored in $\cO(z \log \bar{m}) = \cO(z \log m)$ bits. Then, for any set $\pZ$ of $z$ points, the probability $\pZ$ will be well colored by an $h$ chosen uniformly at random from $\cH$ is 
$\binom{\bar{z}}{z} \frac{z!}{\bar{z}^z} > \frac{\bar{z}^z}{z^z} \frac{z!}{\bar{z}^z} = \frac{z!}{z^z} > e^{-z}$. Thus, if we choose $u$ functions $h \in \cH$ independently and uniformly at random, then the probability that $\pZ$ is not well-colored by at least one of them is at most $(1 - e^{-z})^u$, which is at most $\epsilon$ for $u = e^z\ln(\epsilon^{-1})$. Each such choice can be done in parallel, invoking a separate instance of the procedure in Algorithm~\ref{alg:coverage-streaming}.

Alternatively, we can obtain a deterministic algorithm by making use of a \emph{$z$-perfect family} $\cH$ of hash functions from $\pU \to [\bar{z}]$. Such a family has the property that for any subset $\pZ \subseteq \pU$ of size at most $z$, some function $h \in \cH$ is injective on $\pZ$. Schmidt and Siegal~\cite{SchmidtSiegalS90} give a construction of such a family in which each function can be specified by $\cO(\bar{z}) + 2\log\log |\pU|$ bits.\footnote{There have been several subsequent improvements obtaining smaller families $\cH$ of $z$-perfect hash functions (e.g.~\cite{AlonYusterZwick95,DBLP:conf/focs/NaorSS95,DBLP:conf/soda/ChenLSZ07}). For simplicity, we use the result of~\cite{SchmidtSiegalS90}, which gives explicit bounds on the space required for storing and computing such functions and suffices to obtain poly-logarithmic space in our setting.} Thus, we can simply run our streaming algorithm in parallel for each of the $2^{\cO(\bar{z})}\log^2(m)$ such functions.

Combining the above observations, we have the following:
\begin{theorem}\label{thm:streaming-coverage}
Let $\cM$ be an $\ell$-matchoid and $z \in \posints$. For any  $\epsilon > 0$, there is a randomized streaming algorithm that succeeds with probability $(1-\epsilon)$ and computes a kernel $R$ for \textsc{Maximum $(\cM,z)$-Coverage}. Moreover, $|R| \leq (4e)^{z}\Gamma_{\ell,z}\ln(\epsilon^{-1})$, where $\Gamma_{\ell,z} \triangleq \sum_{q = 0}^{(z-1)\ell}\ell^q$. At all times during its execution, the algorithm stores at most $|R|+1$ sets of at most $\cO(z)$ points each and requires at most $\cO(z\ell\log(n) + (4e)^z\ln(\epsilon^{-1})z\log(m))$ additional bits of storage. For $\ell = 1$, we have $|R| \leq (4e)^zz\ln(\epsilon^{-1})$  and for $\ell > 1$, $|R| = \cO\bigl((4e)^{z}\ell^{(z-1)\ell}\ln(\epsilon^{-1})\bigr)$.

There is also a deterministic algorithm producing a kernel $R$ for the same problem with $|R| \leq 2^{\cO(z)}\Gamma_{\ell,z}\log^2(m)$. At all times during its execution, it stores at most $|R|+1$ sets of at most $\cO(z)$ points each and uses at most
$\cO(z\ell\log n) + 2^{\cO(z)}z\log^2(m)\log\log(m)$  additional bits of storage. For $\ell = 1$, we have $|R| = 2^{\cO(z)}z\log^2(m)$  and for $\ell > 1$, $|R| = 2^{\cO(z)}\ell^{(z-1)\ell}\log^2(m)$.
\end{theorem}
\begin{proof}
Let $O$ be an optimal solution for the problem, and let $\pZ$ be the set of  up to $z$ points of maximum weight covered by $O$. For the randomized algorithm, we process each element of the input stream with $e^z\ln(\epsilon^{-1})$ parallel executions of the procedure \textsc{StreamingMaxCoverage} from Algorithm~\ref{alg:coverage-streaming}, each with a function $h$ sampled uniformly and independently at random from the described $\bar{z}$-wise independent family $\cH$. We then let $R$ be the union of all the sets produced by these processes. With probability at least $(1-\epsilon)$, $\pZ$ is well-colored by one such $h$. Consider the process \textsc{StreamingMaxCoverage} corresponding to this choice of $h$ and let $R$ be its output. For 
every $C \subseteq [\bar{z}]$, all elements $e  \in X_C$ will be processed by a procedure $\SRepSet_{C}$ in this instance. By Theorem~\ref{thm:main-streaming}, each process $\SRepSet_C$ used in  \textsc{StreamingMaxCoverage} then produces a joint $z$-representative set $R_C$ for $(X_C,\cM,w_C)$. Thus, by Lemma~\ref{lem:coverage-repset}, the output $R = \bigcup_{C \subseteq [\bar{z}]}R_C$ for this procedure is a kernel for \textsc{Maximum $(\cM,z)$-Coverage}.

In total, the algorithm maintains $2^{\bar{z}}e^{z}\ln(\epsilon^{-1}) \leq (4e)^z\ln(\epsilon^{-1})$ procedures $\SRepSet_C$. By Theorem~\ref{thm:main-streaming}, each such procedure returns a set of $R_C$ containing at most $\Gamma_{\ell,z} \triangleq \sum_{q=0}^{(z-1)\ell}\ell^q$ elements. For each element, we discard all but the heaviest point of each color class. Thus $|R| \leq (4e)^{z}\Gamma_{\ell,z}\ln(\epsilon^{-1})$ and for each element of $R$, we must store at most $\bar{z} = \cO(z)$ points. When a new element arrives, we can perform the updates in each procedure sequentially, temporarily storing at most one element and using at most $\cO(z\ell\log n)$ bits of additional storage, as shown in Theorem~\ref{thm:main-streaming}. Additionally, we must store $\cO(z\log m)$ bits for the hash function in each of the $(4e)^{z}\ln(\epsilon^{-1})$ procedures \textsc{StreamingMaxCoverage}. Thus, the total number of additional bits required is at most $\cO\bigl(z\ell\log(n) + (4e)^{z}\ln(\epsilon^{-1})z\log(m)\bigr)$.

For the deterministic algorithm, we proceed in the same fashion, but instead use each of the $2^{\cO(\bar{z})}\log^2(m) = 2^{\cO(z)}\log^2(m)$ functions in the $z$-perfect hash family $\cH$, each of which requires at most $\cO(z + \log\log(m))$ bits to store. Then, $\pZ$ will be well-colored by at least one of these functions. By a similar argument as above, the union $R$ of the $2^{\cO(z)}\log^2(m)$ procedures \textsc{StreamingMaxCoverage} will then be a kernel. By a similar calculation, $|R| \leq 2^{\cO(z)}\Gamma_{\ell,z}\log^2(m)$ and the total number of additional bits required is at most $\cO(z\ell\log n) + 2^{\cO(z)}z\log^2(m)\log\log(m)$.
\end{proof}

In Theorem~\ref{thm:streaming-coverage}, we have stated our results in the streaming setting where the primary concern is the space used by the algorithm. However, we note that our algorithms also translate directly to fixed-parameter tractable algorithms for the offline setting, in which the primary concern is computation time. Specifically, instead of processing elements in a stream using multiple instances of $\SRepSet$ we can simply execute multiple instances of the offline procedure $\RepSet$. Combining Theorem~\ref{thm:matchoid-offline} with our analyses from the streaming setting then immediately gives the following.
\begin{theorem}
\label{thm:fpt-algos}
There are fixed-parameter tractable algorithms computing a kernel $R$ for \textsc{Maximum $(\cM,z)$-Coverage} requiring a number of independence oracle calls proportional to $n$ times the stated upper bounds on $|R|$ in  Theorem~\ref{thm:streaming-coverage} plus the time required to sort the input by weight for each of the $(4e)^{z}\ln(\epsilon)$, or $2^{\cO(z)}\log^2(m)$ representative sets maintained, respectively.
\end{theorem}

\bibliographystyle{plain}
\bibliography{library,names}

\appendix

\section{Hardness Results for Alternative Parameterizations}
\label{sec:hardn-results-altern}

Here, we provide some justification for our choice of parameters for each of the problems we consider by showing that the problems become hard if we use any strict subset of the parameters proposed.

First, we note that the 3-dimensional matching problem, which the 3-matchoid problem generalizes, is one of Karp's original NP-hard problems~\cite{DBLP:conf/coco/Karp72}. It follows that all of our problems remain NP-hard when parameterized by $\ell$ alone.

For linear objectives, we parameterize by $\ell$ and $k$. Here, we note that if we parameterize by $k$ alone, we can encode an arbitrary instance of the \textsc{Independent Set} problem, where $k$ is the size of the independent set. This problem is known to be $W[1]$-hard~\cite{10.5555/2568438}.  Given an arbitrary graph $G=(V,E)$, we let our ground set $X$ be $V$, and use an unweighted objective that sets $w(e) = 1$ for each $e \in V$. Then, we introduce a uniform matroid of rank $1$ on $\{u,v\}$ with each edge $(u,v) \in E$. Note that some $S \subseteq X = V$ is then independent in all matroids if and only if no pair of vertices in $S$ share an edge. Moreover, we have a solution of value at least $k$ for our problem if and only if we can select $k$ elements from $S$ and so have an independent set of size $k$ in $G$.


For coverage functions, parameterizing by the number $k$ of elements chosen immediately gives the \textsc{Maximum $k$-Coverage} problem, which is $W[2]$ hard~\cite{DBLP:journals/ita/BonnetPS16}. Here, we parameterize instead by the number of points $z$ that are covered and $\ell$. If instead we parameterize by only $z$, we can again encode an arbitrary instance of \textsc{Independent Set} as described above. We encode our unweighted objective by letting each element of $X$ cover a single, unique point. Then, similar to the discussion for the case of linear functions, we have an independent set of size $z$ in $G$ if and only if we have a set of elements that is independent in all our matroids covering $z$ points.

\end{document}